\documentclass[12pt]{article}

\usepackage[onehalf]{paper-arxiv}
\usepackage{siunitx}   % for units

\hypersetup{
    pdftitle={The Gatekeeping Expert's Dilemma},
    pdfkeywords={Gatekeeper, Dilemma, Disclosure, Audit, Financial Reporting, Strategic Communication, Bargaining, Veto Bargaining, Transparency, Expertise, Information Design, Voluntary Disclosure},
    pdfsubject={The Gatekeeping Expert's Dilemma by Shunsuke Matsuno (University of Chicago Booth School of Business)},
    pdfauthor={Shunsuke Matsuno}
}

\newcommand{\Xl}{\underline{X}}
\newcommand{\Xh}{\overline{X}}
\newcommand{\Mprec}{\mathcal{M}^{\mathrm{p}}}     % silence set
\newcommand{\Mvag}{\mathcal{M}^{\mathrm{v}}}       % silence set
\newcommand{\nd}{\mathcal{X}_0}     % silence set
\newcommand{\info}{\mathrm{Info}}     % Amount of information
\newcommand{\self}{D^\mathrm{self}}   % Self signaling set

\NewDocumentCommand{\ab}{m G{a} G{b}}{%  
  \LeftIndex[]_{#2}{#1}_{#3}%
}

\declaretheorem[name=Lemma]{lemmap}
\zcRefTypeSetup{lemmap}{
  name-sg=lemma,
  Name-sg=Lemma
}

\graphicspath{{figures/}}

\title{The Gatekeeping Expert's Dilemma}

\author{Shunsuke Matsuno%
\thanks{University of Chicago Booth School of Business (email: \href{mailto:Shunsuke.Matsuno@chicagobooth.edu}{\tt{Shunsuke.Matsuno@chicagobooth.edu}}).
First Draft: October 2025.
I thank my dissertation committee members, Cyrus Aghamolla, Tim Baldenius, Jon Glover, Ilan Guttman, and Amir Ziv for their invaluable guidance and support.
For helpful comments and discussions, I thank Daniel Aobdia, Li Azinovic-Yang, Jeremy Bertomeu, Matthias Breuer, Hans Christensen, Nicolas Inostroza, Jeffrey Jou, Anthony Le, Rongchen Li, Lisa Liu, Iván Marinovic, Valeri Nikolaev, Haresh Sapra, Syrena Shirley, Jack Stecher, Sang Wu, Kun Zhang and 
seminar participants at various conferences and workshops.
}
\thanksnosymbol{Supplementary material that contains additional details and proofs is available on the author's website (\url{https://shunsukematsuno.github.io}) [\href{https://shunsukematsuno.github.io/papers/supplementary-material_Gatekeeping-Experts-Dilemma.pdf}{supplementary material}].}
}

\date{November 23, 2026}

\begin{document}

\pagenumbering{gobble}   % Suppress page numbering for the title page

\begin{titlingpage}
    \maketitle

    \vspace{-2em}    % The titleign package adds certain space after Date.
    
    \begin{center}
    \end{center}

    \vspace{1em} 

    \begin{abstract}
        This paper studies how experts with veto power---gatekeeping experts---influence agents through communication.
        Their expertise informs agents' decisions, while veto power provides discipline. 
        Gatekeepers face a dilemma: transparent communication can invite gaming, while opacity wastes expertise.
        How can gatekeeping experts guide behavior without being gamed?
        Many economic settings feature this tradeoff, including bank stress tests, environmental regulations, and financial auditing.
        Using financial auditing as the primary setting, I show that strategic vagueness resolves this dilemma: by revealing just enough to prevent the manager from inflating the report, the auditor guides the manager while minimizing opportunities for manipulation.
        This theoretical lens provides a novel rationale for why auditors predominantly accept clients' financial reports.
        Comparative statics reveal that greater gatekeeper independence or expertise sometimes dampens communication.
        This paper offers insights into why gatekeepers who lack direct control can still be effective.
    \end{abstract}

    \vspace{1em}
    \begin{description}
        \item[Keywords:] Gatekeeper, Strategic Communication, Auditing, Bargaining, Transparency
        \item[JEL Classification:] C72, D82, D83, G30, M42
    \end{description}

\end{titlingpage}

\newpage

\pagenumbering{arabic}   % Restart page numbering for the main document
\setcounter{page}{1}

\section{Introduction}
Gatekeepers play a central role in many economic settings.
They are experts who have the option to withhold support or access to critical resources \citep{kraakmanGatekeepersAnatomyThirdParty1986}.
Examples abound: a credit rating agency assesses whether a bond's issuer is risky; external auditors sign off on firms' financial statements; environmental regulators grant permits or certify compliance; proxy advisors issue recommendations that can effectively block voting proposals; and an academic journal editor decides whether a paper is to be published. 
Gatekeepers cannot dictate what the agent ought to do;
they can only communicate what they believe the agent ought to do.
However, communication is a fragile tool.
If a rating agency reveals exactly how it assesses issuers, they may game the standard instead of reducing risk.
A completely opaque gatekeeper is no better.
If an environmental regulator never explains what constitutes compliance, firms cannot improve their practices even when willing.
How then can gatekeepers, constrained by a lack of direct control and the precarious nature of communication, effectively influence agents?

This paper studies a model of a \textit{gatekeeping expert}---an expert who influences an agent through veto power and communication.
I study the extent to which the gatekeeping expert can influence the agent and how she does so.
To fix ideas, I use financial auditing as the primary setting.
Auditors are a quintessential example of gatekeeping experts \citep{coffeeGatekeeperFailureReform2004}.
They use their expertise to monitor firms' accounting reports and certify whether the reports are free from material misstatements \citep{ronenCorporateAuditsHow2010}.

In my model, an auditor decides whether to accept a firm manager's financial report.
\footnote{I use male pronouns (``he'') for the manager and female pronouns (``she'') for the auditor.}
The auditor wishes to ensure that the manager's report complies with accounting standards.
These standards are often ambiguous and require interpretation, especially for complex business transactions \citep{secFinalReportAdvisory2008}.
The auditor has expertise in interpreting and applying these complex standards \citep{bonnerExperienceEffectsAuditing1990}.
The auditor's \textit{preferred report} is the treatment the auditor deems appropriate, given her expertise and all available information.

The firm's manager has misaligned incentives: he seeks solely to maximize the reported value.
For any particular transaction, once the audit has been conducted, the manager possesses no private information beyond what the auditor observes.
I make these stark assumptions to delineate the gatekeeping expert's communication problem.
The framework can serve as a benchmark for further studies that incorporate richer institutional features.

The auditor accepts only reports that are close enough to the preferred report and rejects others.
The auditor's \textit{independence}, which captures the cost she incurs from rejecting a report, determines the margin of error she tolerates. 
A more independent auditor tolerates smaller deviations from her preferred report.

I say that the manager \textit{games} the auditor when the manager successfully chooses the highest report the auditor tolerates.
Gaming is an information problem rooted in the auditor's lack of direct control.
The manager can game the auditor when he knows how much over-reporting the auditor will tolerate, because the auditor cannot dictate the manager's report.

The auditor communicates her expertise before the manager proposes a report.
I consider two modes of communication: \textit{precise} and \textit{vague}.
With precise communication, the auditor tells the manager the exact value of the preferred report.
With vague communication, she communicates a range of possible values.
The auditor's message must be truthful: it cannot contradict her knowledge about the preferred report, which is determined by relevant accounting standards.
The auditor cannot commit ex ante to how she will communicate or which reports she will accept.

Auditor communication is a critical part of the auditing process.
During the early years of the Sarbanes-Oxley Act, critics argued that the Act was diminishing financial reporting quality, as auditors no longer provided sufficient guidance on complex accounting issues.
The regulatory body responded by stressing that auditors are allowed to provide guidance \citep{goelzerCostsBenefitsSarbanesOxley2005}.
The episode highlights the importance of how expert auditors communicate with managers.
\footnote{Because auditor--client communications are confidential, direct evidence on how auditors communicate with managers is scarce.
In this regard, \cite{beattieClosedDoorsWhat2001} provide interesting case studies based on matched interviews with auditors and managers.}

In my model, the auditor's communication shapes how the bargaining unfolds.
Without communication, the manager's report is sometimes unacceptable to the auditor.
The auditor thus would like to communicate her expertise, especially when she anticipates that the manager's report will deviate significantly from the preferred report.

But here is the dilemma: if the auditor communicates her expertise precisely, she will be gamed by the manager.
This is because revealing the exact value of the preferred report also reveals the exact upper bound of what she is willing to accept.
The manager then exploits this information.
He proposes the highest possible value that the auditor is indifferent between accepting and rejecting.
This is what I call the \textit{gatekeeping expert's dilemma}: the auditor would like to communicate her expertise, but if she does, the manager can game the auditor and render her control minimally effective.

The auditor does not always prefer complete silence, however, as the strategic value of silence depends on when she speaks.
Consider a communication strategy where the auditor reveals the preferred report only for certain values and remains silent for others.
Her silence then becomes informative.
It signals to the manager that the preferred report falls outside the range where she would have spoken. 
From the auditor's silence, the manager learns something about the preferred report, but he still does not know the exact value.
This strategic silence allows the auditor to prevent gaming while still steering the manager's report closer to her preferred outcome. 
However, a crucial limitation exists: whenever the auditor reveals the preferred report, she is gamed and her payoff is minimized.
The auditor's ability to leverage silence is thus constrained by the cost she incurs when she speaks.
I show that the auditor can sometimes, but not always, strictly improve her payoff by carefully choosing what not to say.

This motivates the question: can the auditor utilize vague language to communicate her expertise while minimizing the risk of gaming?
The main result of this paper shows that she can, and it characterizes how she does so.
The \textit{maximal acceptance} property defines an auditor-optimal equilibrium: 
across all equilibria, an auditor-optimal equilibrium maximizes the probability that the manager's report is acceptable.
For each message, the auditor steers the manager's beliefs in the right direction while successfully eliciting reports close to her preferred one.
Vague talk lets the auditor guide the manager without tipping her hand.

In practice, auditors accept the vast majority of reports they review.
For example, in the U.S., over $99\%$ of public firms' financial statements receive an unqualified opinion---an auditor's approval that the financial statements are accurate \citep{ciprianoHasLackUse2017}. 
My model suggests that this high acceptance rate may, in part, reflect auditors' use of strategic vagueness to guide managers' reporting behavior.

Comparative statics highlight subtle interactions between the auditor's communication strategy and the manager's reporting incentives.
I study the \textit{amount of information} the auditor provides in equilibrium, defined as the reduction in the uncertainty about the preferred report.

First, I examine the effect of auditor independence.
Without communication, an auditor with high independence would end up rejecting often, since she tolerates only a narrow margin of error.
Therefore, one might conjecture that a high-independence auditor on average provides more information to ensure acceptable reports.
Somewhat surprisingly, this is not necessarily true: when the manager's incentive to inflate the report is strong, a more independent auditor may provide \textit{less} information (vaguer communication).

Intuitively, vaguer messages give the manager more room to ``gamble'' on a higher report that may be rejected.
For a fixed level of independence, this incentive is strongest when the auditor's preferred report is low, as the manager has less to lose by gambling.
As a result, the auditor chooses to be vaguer over regions where her preferred reports are high than where they are low.
Now, consider what happens when the auditor becomes more independent.
She must be more precise to ensure that the reports meet her tighter acceptance criteria; this effect is particularly pronounced for low realizations of preferred reports, where the manager's incentive to gamble is strong.
This, in turn, lets the auditor be less precise when her preferred report is high.
If the increase in vagueness at high realizations outweighs the increase in precision at low realizations, a more independent auditor will convey less information overall.

Second, I study how information asymmetry---measured by the prior variance of auditor-preferred reports---affects communication.
I interpret this variance as transaction complexity, or equivalently, the importance of auditor expertise.
Some transactions are inherently more complex than others and require greater expertise for the auditor to determine her preferred report (e.g., the sale of a division with future contingencies as opposed to a simple spot sale of inventory).
Greater transaction complexity would translate into greater ex-ante uncertainty about the auditor's preferred report.
I ask whether the auditor reveals more information relative to the increase in transaction complexity.
The analysis again reveals a non-monotonic relationship.
When the manager is already inclined to inflate the report, added complexity amplifies that tendency.
The auditor must provide more information to rein in the manager (``strategic effect'').
At the same time, greater complexity increases the probability that the auditor's preferred report is high, in which case the auditor is relatively vaguer than when the preferred report is low (``statistical effect'').
The effect of increased complexity on communication depends on which force dominates.

The findings have policy implications for the regulatory interventions of gatekeepers.
A popular belief is that gatekeepers should remain independent.
In financial auditing, auditor independence is a particularly important issue \citep{mautzPhilosophyAuditing1961}.
Since the Sarbanes-Oxley Act, regulators have intensified efforts to strengthen auditor independence.
Yet my model indicates that greater auditor independence may lead to less information production.
Similarly, expanding the scope for auditor judgment (i.e., increasing the prior variance over auditor-preferred reports) does not necessarily lead to more auditor guidance.
\footnote{For example, moving from bright-line rules to principles-based standards would expand the scope for auditor judgment.}
My analysis underscores that the equilibrium consequences of such policy interventions may be counter-intuitive and potentially at odds with the intended regulatory objectives.

I intentionally leave out many institutional features of auditing to isolate the core of the gatekeeping expert's dilemma.
I briefly discuss several extensions in \zcref{sec:conclusion}.
For example, I discuss how to extend the model to allow for endogenous choice of transactions and thus have real effects.
Since the manager prefers to choose the transaction that is ``gameable,'' a communication strategy that is auditor-optimal for a fixed transaction may induce the manager to distort the transaction choice.\footnote{In my framework, the manager does not possess private information, and the gatekeeper's guidance on accounting standards creates real effects---how the manager structures the transaction. In contrast, in the real-effect models of \cite{kanodiaRealEffectsPerspective2016}, the manager has private information, and the stock market's reaction given accounting standards creates real effects.}

The model's logic extends beyond auditing to other gatekeeping contexts.
\footnote{Even within auditing, the model applies to non-financial audits---most notably sustainability audits, such as when firms report greenhouse gas emissions. 
The absence of established standards makes auditors' judgment even more critical than in financial audits \citep{enriquesGreenGatekeepers2024}.}
Consider, for instance, bank stress tests.
The Federal Reserve conducts periodic tests to assess banks' financial health.
Banks make operational decisions, and the regulator uses its expertise to determine whether a bank passes the test.
The regulator faces the gatekeeping expert's dilemma: if it reveals too much about how it evaluates banks, banks may game the test rather than improve their financial health.
This issue of ``model secrecy'' has drawn attention from both academics and practitioners \citep{leitnerModelSecrecyStress2023,barrSpeechGovernorBarr2025}.
My framework provides a lens to understand the Fed's tradeoff and illustrates how the regulator can shape bank behavior by strategically choosing secrecy.
In \zcref{appsec:other-applications}, I elaborate on the stress test example and discuss other applications, including environmental regulation, capital budgeting, and patent examination.

\subsection{Related Literature}
The legal literature on gatekeepers focuses primarily on gatekeeper liability \citep{kraakmanGatekeepersAnatomyThirdParty1986,coffeeGatekeeperFailureReform2004,tuchLimitsGatekeeperLiability2017}.
As \cite{kraakmanGatekeepersAnatomyThirdParty1986} insightfully points out, such liability becomes necessary only when alternative enforcement mechanisms are unavailable and when gatekeepers can effectively utilize their expertise and prevent agents from engaging in undesirable behavior.
However, the effectiveness of gatekeepers is often unclear, precisely because they do not control agents directly.
My study contributes to this literature by showing that communication, despite the risk of inviting gaming, is essential for gatekeepers to be effective.

By modeling auditors as gatekeeping experts and focusing on their  problem, my paper contributes to the theoretical literature on auditing.
A major strand of the auditing literature studies the role of auditors in principal-agent settings, where auditors verify the agent's report \citep[``Costly State Verification'' models; see, e.g.,][]{townsendOptimalContractsCompetitive1979,borderSamuraiAccountantTheory1987,mittendorfRoleAuditThresholds2010}.
\footnote{Another strand of the literature analyzes how auditors conduct audits strategically in otherwise decision-theoretic frameworks that are complemented by strategic managers \citep[``Strategic Auditing'' models; see, e.g.,][]{fellinghamStrategicConsiderationsAuditing1985,shibanoAssessingAuditRisk1990}.
Other work in auditing has analyzed 
auditor switching \citep{dyeInformationallyMotivatedAuditor1991,liInformationbasedTheoryAuditor2023}, 
auditor independence \citep{deangeloAuditorIndependenceLow1981,antleAuditorIndependence1984,aryaAuditorIndependenceRevisited2014}, 
auditors' legal liability \citep{dyeAuditingStandardsLegal1993,lauxAuditorLiabilityClient2010},
and auditor--client matching \citep{liTwoSidedMatchingAudit2025,achimTheoryAuditorSelection2025}.
\cite{yeTheoryAuditingEconomics2023} surveys the theoretical literature on auditing.}
In contrast, my paper focuses on the auditor's expertise and her communication problem, moving beyond the traditional view of the auditor as an information verifier.

Despite the importance of auditor--client negotiation \citep{gibbinsEvidenceAuditorClient2001,nelsonEvidenceAuditorsManagers2002,brownNegotiationResearchAuditing2008,lennoxSurveyArchivalAudit2025}, its theoretical treatments remain sparse.
The work of \cite{antleConservatismAuditorClientNegotiations1991} is a notable exception. 
They model the audit as a negotiation where the client has superior information about the firm's performance. 
I complement their research by considering the case in which the auditor has superior information and focusing on the information-sharing aspect of the negotiation.

I focus on communication strategies that are truthful, in line with \cite{milgromGoodNewsBad1981}.
\footnote{\cite{crawfordStrategicInformationTransmission1982} analyze non-truthful communication (``cheap talk'') and show that a biased sender coarsens information to achieve credible communication. This result resembles the gatekeeper's optimal vague communication in my model. However, the underlying economic forces differ fundamentally, beyond the distinction between the communication technologies.
In their model, the sender wants to convey as much information as possible but lacks credibility; in my model, the gatekeeping expert does not want to reveal too much information due to the risk of being gamed.}
The distinction between precise and vague communication follows the work of \cite{hagenbachSimpleRichLanguage2017}.
In my model, the auditor's vague communication strategy is a partitional communication strategy.
\cite{glodeVoluntaryDisclosureBilateral2018} and \cite{aliVoluntaryDisclosurePersonalized2023} also study truthful partitional disclosure in a monopolistic screening setting. 
These papers derive buyer-optimal communication strategies. 
My setting features different economic forces: the sender (gatekeeper) would like to induce a report that is close to the preferred report, while the buyer in \cite{glodeVoluntaryDisclosureBilateral2018} and \cite{aliVoluntaryDisclosurePersonalized2023} would like to induce the lowest possible price.
Consequently, precise communication and silence can strictly improve the gatekeeper's payoff, while the buyer's payoff under such a strategy never improves in their models.
\footnote{Another important difference is that, in my model, even if we relax the assumption of truthful communication, the gatekeeper can still credibly convey some information. By contrast, in \cite{glodeVoluntaryDisclosureBilateral2018} and \cite{aliVoluntaryDisclosurePersonalized2023}, the buyer cannot credibly convey any information if communication is mere cheap talk.}

Communication and veto power are also featured in delegation models \citep{melumadCommunicationSettingsNo1991,alonsoOptimalDelegation2008}.
In these models, an uninformed principal sets a decision rule---a stronger form of \textit{ex-ante} veto power---to influence an informed agent.
In contrast, in my model, an informed gatekeeper uses \textit{ex-post} veto power and communication to influence an uninformed agent.

My paper is also related to the literature on veto bargaining \citep{romerPoliticalResourceAllocation1978}.
Recent papers study veto bargaining under incomplete information \citep{aliSequentialVetoBargaining2023,kimPersuasionVetoBargaining2024}.
I contribute to this literature by showing how gatekeepers (vetoers) can influence the bargaining outcome through strategic communication.

\section{An Example}\label{sec:example}

I illustrate the main idea via a simple example, deferring formal development to later sections.

A manager of a firm and an auditor are in the final, critical negotiation phase of the audit.
They need to resolve how to record the fair value of an illiquid asset.
As an expert, the auditor privately learns an appropriate accounting treatment among those permitted by accounting standards.
She does so by reviewing relevant information, utilizing available resources such as valuation specialists, and relying on her expert judgment \citep{ahnAuditorTaskspecificExpertise2020}.  
Denote by $X$ the auditor's preferred fair value of the asset.
The manager lacks expertise and is uncertain about the fair value the auditor deems appropriate.
The manager's prior belief about $X$ is uniformly distributed over the interval $[1,9]$.
The manager aims to maximize the reported fair value in the financial statements, denoted by $r$.
As a gatekeeping expert, the auditor seeks to ensure that the manager's report closely aligns with her preferred fair value $X$.
If the auditor accepts the manager's report $r$, she incurs a quadratic loss $(r - X)^2$. 
If she rejects the report, she incurs a fixed loss of $1$, while the manager receives a normalized payoff of zero.
Consequently, the auditor accepts any report $r$ within a margin of one from her preferred value $X$.

The focus is the auditor's communication problem: how can the auditor with veto power effectively guide the manager's report through communication?
I structure the analysis around two types of communication strategies: \textit{precise} and \textit{vague}.
If the auditor uses \textit{precise} communication, she either remains silent or simply tells the manager what the fair value is.
In contrast, a \textit{vague} message conveys only that the fair value lies within a specified range of values.
Importantly, the auditor's message must be \textit{truthful}.
If she states that $X=\hat{X}$ (precise communication), then it must be that $X = \hat{X}$.
If she states that $X$ is in the range $[a,b]$ (vague communication), then it must be that $X\in [a,b]$.
One justification for truthfulness is regulatory scrutiny.
If the auditor knowingly provides false information to her client, she risks punishment from regulators.
\footnote{In the U.S., the Public Company Accounting Oversight Board (PCAOB) enforces auditing standards and sanctions auditors for violations.}

\subsection{Precise Communication}
Suppose the auditor communicates her expertise using only precise language or silence.
\footnote{By precise communication, I mean that she reveals $X$ precisely \textit{if} she chooses to communicate about $X$. She may choose to remain silent---in which case, as we will see below, her silence indirectly and vaguely conveys information about $X$.}
The auditor considers how the manager will respond to her communication or the lack thereof.

Absent additional information from the auditor, the manager selects a report $r$ to maximize the expected accepted value.
The auditor accepts any report within the interval $[X-1, X+1]$. 
As $X$ varies from $1$ to $9$, the acceptance region moves from $[0,2]$ to $[8,10]$.
The manager faces a risk-return tradeoff: higher reports are better if accepted, but they are more likely to be rejected.
Any report $r\ge 10$ will be rejected with probability one.
The optimal report solves 
\begin{equation}
    \max_{r\ge 0}\; r\, \prob(|r-X|\le 1).
\end{equation}
The solution, $r_0 = 8$, is inflated compared to the expected value $\E[X]=5$.
When $X\in [7,9]$, the auditor accepts the report $r_0$, but when $X<7$, she rejects it.
The auditor's ex-ante expected loss is therefore 
\begin{equation}
    \begin{bracealign}
        \underbrace{\frac{6}{8}\times 1}_{\text{rejection}} + 
        \underbrace{\int_{7}^{9}(r_{0}-x)^{2}\frac{1}{8}dx}_{\text{acceptance}}
        =\frac{5}{6}\approx 0.83.
    \end{bracealign}
\end{equation}

When $X < 7$, the auditor anticipates that the manager's report will deviate significantly from her preferred fair value.  
Thus, the auditor would like to communicate this to the manager and have him adjust the report accordingly.
Would the auditor benefit from \textit{precisely} communicating the value of $X$?
By revealing her knowledge, the auditor can ensure that the manager does not ``gamble'' by proposing a report likely to be rejected.
However, in doing so, the auditor loses her information advantage.
The manager not only learns that the auditor-preferred fair value is $X$, but also that the auditor will accept any report in the interval $[X - 1, X + 1]$.
Thus informed, the manager inflates the report to $X+1$.
The auditor's payoff is now minimized, as she is indifferent between accepting and rejecting the report---the auditor is \textit{gamed} by the manager.
This is the gatekeeping expert's dilemma.

Under precise communication, the auditor may adopt a strategy of remaining silent when $X$ lies in some range and revealing it precisely otherwise.
When she reveals her preferred fair value, the auditor accepts being gamed; when she remains silent, however, she shifts the manager's beliefs without revealing the exact value of $X$.
For example, consider the auditor's strategy of revealing $X$ only when it is in $[5,9]$.
Silence then implies that $X \in [1,5]$.
The manager is still unsure of $X$, but his report is less likely to deviate substantially from it.
However, as it turns out, the auditor does not benefit from such a strategy under the uniform distribution (\zcref{prop:uniform-precise-ND}).
An auditor-optimal equilibrium outcome under precise communication is attained by always staying silent, which guarantees the auditor a loss of $5/6$.

Can the auditor do better than this?
Next I explore whether the auditor can utilize vague language to communicate her expertise without being gamed by the manager.

\subsection{Vague Communication}
Instead of revealing the exact fair value $X$, the auditor can communicate a range of possible fair values.
Such vague language potentially helps the auditor, because a vague message can guide the manager without revealing the exact upper bound of what the auditor is willing to accept.
When the auditor says that $X$ is in the interval $[a,b]$, the manager learns something (his report becomes better aligned with the auditor's preferred fair value), but he still does not know everything (he cannot propose the auditor's maximal acceptable fair value $X+1$).

To be specific, consider the following vague communication strategy: the auditor reveals whether $X$ is in $[1, 2]$, $[2, 5]$, or $[5, 9]$. 
For each of these messages, the manager solves a new reporting problem.
The optimal such report is
\begin{equation}
    r^* = \begin{cases}
            2 & \text{if } X\in [1, 2],\\
            4 & \text{if } X\in [2, 5],\\
            8 & \text{if } X\in [5, 9].
        \end{cases}
\end{equation}
\zcref{fig:partition-uniform-example} (left panel) illustrates this.
When the manager learns $X\in [5,9]$, the manager's report remains unchanged from the no-communication report $r_0=8$.
Thus, conditional on $X\in [5,9]$, the auditor's expected payoff is identical to the no-communication scenario.
In the region $X<5$, the vague message makes a difference.
When $X < 5$, without communication, the manager's report would always be rejected.
With vague communication, the manager tailors his report to the auditor's message, so it is sometimes acceptable to the auditor.
As a result, overall, the vague communication strategy strictly improves the auditor's expected payoff (and the manager's too).
\footnote{The new expected loss is $7/12\approx 0.58$, which is $30\%$ smaller than the no-communication expected loss of $5/6$.}\textsuperscript{\!,}
\footnote{The constructed strategy indeed constitutes an equilibrium: the auditor does not have an incentive to deviate from this communication strategy.
Suppose the manager holds a ``wishful'' off-path belief: he believes that $X$ is the highest possible value within the off-path message.
The auditor with realization $X$ benefits from an off-path message $[a,b]\subset[1,9]$ only if $|b+1 - X|< 1$.
However, this inequality requires $X>b$, violating the truthful communication requirement.
Thus the manager's wishful off-path belief together with the vague communication strategy constitutes an equilibrium.}

The vague communication strategy described above is just one of many ways to convey the auditor's expertise.
What strategy, then, maximizes the auditor's equilibrium payoff?
The answer is notably simple: the auditor achieves optimal payoff by equally partitioning the support $[1,9]$.
The right panel of \zcref{fig:partition-uniform-example} shows this optimal vague communication.
The shaded area indicates that the manager's report is always within the acceptance region.
The auditor successfully elicits her favorable reports without inefficiently rejecting reports.
The outcome stands in stark contrast to the precise communication case.
With precise communication, in order to ensure acceptable reports, the auditor must reveal all of her expertise and obtains the minimal payoff.

Maximal acceptance is not limited to this particular example but rather is a defining property of an auditor-optimal equilibrium (\zcref{lem:maximum-acceptance-lemma}).
Roughly speaking, when the auditor's communication results in an unacceptable report, the auditor could always convey just enough information to make the manager's report acceptable.
To make this precise, I now turn to a formal description of the model. 

\begin{figure}[htb]
    \centering
    \caption{Vague Communication via Partition}
    \label{fig:partition-uniform-example}
    \scalebox{0.9}{%
        \includegraphics[alt={Partitioning Illustration}]{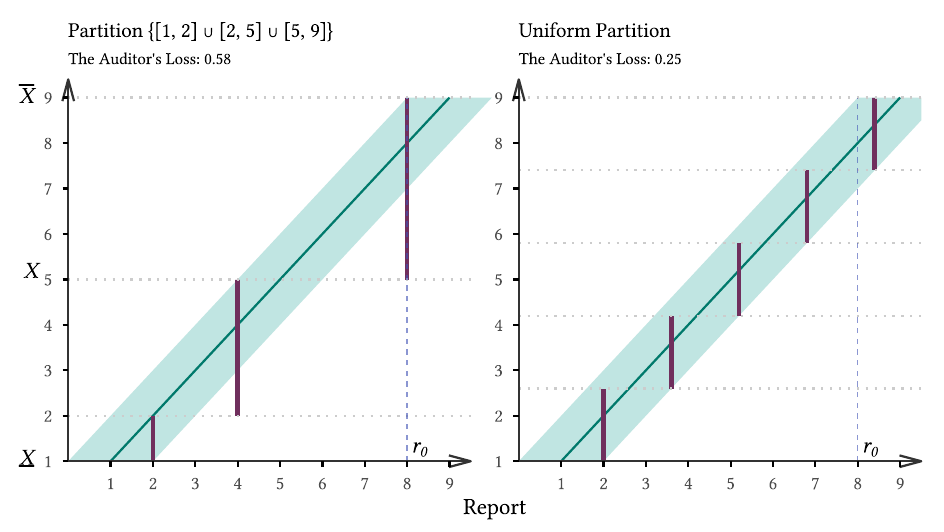}
    }
    \floatfoot{Note: The figure illustrates the auditor's vague communication strategies in the uniform distribution case, $X\sim \mathcal{U}[1,9]$.
    For each value of $X$ in the vertical axis, the range of acceptable reports in the horizontal axis is represented by the shaded region.
    The left panel corresponds to the communication strategy with messages $[1,2]$, $[2,5]$, and $[5,9]$; the right panel to the uniform-partition strategy.
    For each $X$, the corresponding report on the horizontal axis is at the intersection of the horizontal line through $X$ and the thick vertical line.
    For example, in the left panel, when $X=4$, the manager reports $r^*=4$.
    The vertical dashed line at $r_0=8$ indicates the manager's report under no information.}
\end{figure}

\section{A Model of a Gatekeeping Expert}

\subsection{Setup}
A firm's manager and an auditor bargain over how to treat a transaction in the firm's financial statements.
The manager proposes a report $r \in \mathbb{R}$ to the auditor, who decides whether to accept the report.
The auditor privately learns the \textit{preferred report} for the transaction.
The preferred report is represented by a random variable $X$ with support $[\Xl,\Xh]$, where the endpoints may be infinite.
The distribution of $X$ is determined by the parameter $\theta \in \Theta$, which I interpret as transaction/issue characteristics. 
The cumulative distribution function of $X$ given $\theta$ is $F(\cdot\,;\theta)$.
Each distribution in the family $\{F(\cdot\,;\theta)\}_{\theta \in \Theta}$ admits a continuously differentiable, log-concave density function $f(\cdot\,;\theta)$.
This assumption is satisfied by many common distributions, including the uniform, normal, exponential, and logistic distributions \citep{bagnoliLogConcaveProbabilityIts2005}.

The transaction characteristics $\theta$, and thus the distribution of $X$, are common knowledge.
I omit $\theta$ from the notation until the comparative statics analysis.

\subsection{Payoffs}

\subsubsection{The Auditor's Payoff}
The auditor would like to ensure that the manager's report $r$ in the financial statements is close to the preferred value $X$. 
Specifically, if the auditor accepts a report $r$, then her payoff is 
\begin{equation}\label{eq:auditor-payoff}
    u(r,X) = -(r-X)^2.
\end{equation}
The quadratic specification allows for a sharp characterization of equilibrium behavior, and it transparently demonstrates the economic forces in the model.

If the auditor rejects the report, then she obtains an exogenous, commonly known outside option payoff of $-(\pi^R)^2$ with $\pi^R > 0$. 
I refer to $\pi^R$ as the auditor's \textit{independence}: a larger $\pi^R$ corresponds to less independence, as the auditor has more to lose from rejecting a report.

\subsubsection{The Manager's Payoff}
Let $P(r)$ denote the payoff the manager obtains when the report $r$ is accepted.
I take a reduced-form approach and assume only that $P(r)$ is increasing in $r$, without specifying the exact mechanism underlying $P$.
Examples include equity valuation, borrowing terms, bonus triggers, and product-market perceptions.
\footnote{The manager's incentives to boost the reported earnings are extensively documented in the accounting literature \citep{dechowUnderstandingEarningsQuality2010}.}
If the auditor rejects $r$, the reported number becomes a default value $r^R$, which I normalize to zero.

Under this assumption, the manager's choice problem is to push $r$ upward subject to acceptance.
The manager's payoff is therefore
\begin{equation}
    v(r)=\begin{cases}
        P(r) & \text{if }r\text{ is accepted}\\
        0 & \text{otherwise}.
    \end{cases} 
\end{equation}
I adopt the linear normalization $P(r)=r$ in what follows.
This preserves the strategic forces, while simplifying the manager's reporting problem and allowing us to focus on the auditor's communication strategy.

\subsection{Communication}
After observing $X$ but before the manager proposes the report $r$, the auditor can communicate her expertise.
Let $\mathcal{M}(X)$ be the set of messages available to the auditor with realization $X$.
Each message $D\in \mathcal{M}(X)$ is a (Borel) subset of $[\Xl,\Xh]$.
The auditor's strategy is a mapping $\sigma:[\Xl,\Xh]\to \mathcal{M}(X)$.

The auditor's communication is \textit{truthful}: her message must be consistent with $X$.
As in the example, I consider \textit{precise} and \textit{vague} communication.
The auditor's message space under precise communication is given by
\begin{equation}\label{eq:precise-message-def}
    \Mprec(X) \coloneq \{X, \varnothing\},
\end{equation}
where $\sigma(X)=\varnothing$ denotes silence. 
\noeqref{eq:precise-message-def}     % ! Force to show the equation number even when it is not referenced. This is referenced in a separate document, so we need to put the number.

Under vague communication, the set of available messages is given by
\footnote{I rule out precise messages from the vague message space, as $a\ne b$ in \eqref{eq:vague-message-def}.
This is purely for the sake of exposition and without loss;
the auditor's best equilibrium payoff remains unchanged when the message space is expanded from $\Mvag$ to $\Mprec \cup \Mvag$.}
\begin{equation}\label{eq:vague-message-def}
    \Mvag(X) \coloneq \{[a,b] \subset [\Xl,\Xh] \mid a< b, [a,b]\ni X \}.
\end{equation}

\subsection{Strategies, Timeline, and Equilibrium Concept}
\subsubsection{The Reporting Problem with the Gatekeeper}
After observing the auditor's message $D$, the manager updates his belief according to a mapping $B:\mathcal{M}\to \Delta([\Xl,\Xh])$, where $\Delta([\Xl,\Xh])$ is the set of probability measures on $[\Xl,\Xh]$.
The manager's reporting strategy is a mapping $r: \Delta([\Xl,\Xh])\to \mathbb{R}_+$.
I slightly abuse notation and write $r(D)$ to denote the manager's report when the posterior belief is $B(D)$.

The auditor either accepts or rejects the report $r$.
The auditor's acceptance strategy is denoted by $\alpha:\mathbb{R}_+\times [\Xl,\Xh] \ni(r,X)\to \alpha(r,X)\in \{0, 1\}$, where $\alpha(r,X)=1$ means acceptance.

The manager solves the following \textit{reporting problem}:
\begin{equation}\label{eq:reporting-problem-stock-price}
    \max_{r\ge 0} \E[v(r) \mid D] = \max_{r\ge 0} \E[P(r) \alpha(r,X) \mid X\in D].
\end{equation}
He maximizes the payoff from the accepted report times the probability of acceptance, given his posterior belief after observing the auditor's message $D$.

\subsubsection{Timeline}
To summarize, the timeline of the game is as follows:
\begin{enumerate}
    \item The auditor privately learns her preferred report $X\sim F(\cdot;\theta)$ for the transaction $\theta$.
    \item The auditor communicates about $X$ to the manager.
    \item The manager observes the auditor's message and proposes a report $r$.
    \item The auditor decides whether to accept $r$.
    \item The payoffs are realized.
\end{enumerate}

\subsubsection{Equilibrium Concept}
I employ Perfect Bayesian Equilibrium (PBE) as the equilibrium concept.

\begin{definition}\label{def:PBE}
    The tuple $\langle\sigma, \alpha, r, B\rangle$ is a \textit{Perfect Bayesian Equilibrium} (PBE) if the following conditions hold:
    \begin{enumerate}
        \item The auditor's strategy ($\sigma, \alpha$) maximizes her expected payoff for any $X$ and $r$.
        \item The manager's reporting strategy $r$ solves the reporting problem \eqref{eq:reporting-problem-stock-price}.
        \item The manager's belief updating rule $B$ follows Bayes' law for on-path messages.
    \end{enumerate}
\end{definition}

After observing $X$, the auditor chooses a message $\sigma(X)\in \mathcal{M}(X)$, where the message space is either precise ($\Mprec(X)$) or vague ($\Mvag(X)$).
For $\sigma$ to be part of a PBE, the auditor must not benefit from deviating to any other message $D\in \mathcal{M}(X)$.

My primary focus is how the auditor shares her expertise.
Thus, I focus on the \textit{auditor-optimal equilibrium}, defined as an equilibrium that maximizes the auditor's expected payoff across all possible equilibria.

In the appendix, I show that a common equilibrium refinement provides strong support for an auditor-optimal equilibrium (\zcref{prop:app_GPFE,prop:app_GPFE_converse}).
The refinement, originally developed by \cite{grossmanPerfectSequentialEquilibrium1986} and \cite{farrellMeaningCredibilityCheaptalk1993}, is commonly used in truthful communication games \citep[e.g.,][]{bertomeuVerifiableDisclosure2018,glodeVoluntaryDisclosureBilateral2018,aghamollaMandatoryVsVoluntary2024}.

\subsection{The Acceptance Decision and Terminology}
The auditor accepts the report $r$ if and only if the acceptance payoff is greater than the rejection payoff:
\begin{equation}
    \alpha(r,X) = 1 \iff u(r,X) \ge -(\pi^R)^2.
\end{equation}
From \eqref{eq:auditor-payoff}, this condition reduces to $|r-X|\le \pi^R$.
Hence the auditor accepts any report $r \in [X - \pi^R, X + \pi^R]$ and rejects all others.

Call a report $r$ \textit{acceptable} if $|X-r|\le \pi^R$.
A report such that $|X-r|< \pi^R$ gives the auditor a payoff strictly greater than $-(\pi^R)^2$.
I call this a $\textit{favorable}$ report.
Say that the auditor is \textit{indifferent} if she receives a report $r$ such that $|r-X|=\pi^R$, in which case she receives a payoff of $-(\pi^R)^2$.
The auditor wishes to induce a favorable report to obtain a payoff strictly greater than $-(\pi^R)^2$.
I assume $\Xh - \Xl > 2\pi^R$ in the case of bounded support.

\section{The Reporting Problem for Arbitrary Intervals}\label{sec:RP-General}
The auditor's communication strategy is shaped by how the manager reacts to the auditor's message.
I therefore begin with the manager's reporting problem after he learns $X \in [a,b]$.
Denote by $\ab{F}$ the cumulative distribution function of $X$ truncated over the interval $[a,b]$.
Define the acceptance probability function by $A(r) \coloneq \ab{F}(r + \pi^R) - \ab{F}(r - \pi^R)$.
It represents the probability that the report $r$ falls in the acceptance region $[X - \pi^R, X + \pi^R]$ given that $X\in [a,b]$.
The reporting problem is to find a report $r^*$ that maximizes the manager's expected payoff:
\begin{equation}\tag{$\mathrm{Report}(a,b)$}\label{eq:reporting-problem}
    \max_{r \ge 0}\; r \, A(r). 
\end{equation}

The following lemma assures that \ref{eq:reporting-problem} always admits a unique solution.

\begin{lemma}\label{lem:reporting-problem-solution}
    For any $[a,b]\subset[\Xl,\Xh]$, \ref{eq:reporting-problem} admits a unique solution $r^* \ge a + \pi^R$.
\end{lemma}

The length of the interval matters.
If $b-a>2\pi^R$, for any report $r$, there is always some chance that $r$ falls outside the auditor's acceptance region.
If on the other hand $b-a\le 2\pi^R$, the manager can always choose a report in the middle to guarantee that the auditor accepts it.
The acceptance probability in this ``short-interval'' case is
\begin{equation}\label{eq:acceptance-probability}
    A(r)
    =\begin{cases}
        \ab{F}(r+\pi^{R}) & r\in[a-\pi^{R},b-\pi^{R})\\
        1 & r\in[b-\pi^{R},a+\pi^{R}]\\
        1-\ab{F}(r-\pi^{R}) & r\in(a+\pi^{R},b+\pi^{R}].
    \end{cases}
\end{equation}
Starting from $r = a - \pi^R$, the acceptance probability \eqref{eq:acceptance-probability} increases with $r$. 
When $r$ is small, the manager worries that $X$ is actually too large.
This would risk the report being too low to be accepted. 
Once $r$ reaches $b-\pi^R$, the acceptance probability becomes one and stays constant until $r = a + \pi^R$.
Beyond this point, $A(r)$ decreases monotonically: the manager risks the possibility that the report is too high to be accepted.
In \zcref{fig:acceptance-prob-normal}, I plot the acceptance probability \eqref{eq:acceptance-probability} for a normal distribution.

Since $A(r)$ has a kink at $r=a+\pi^R$, the first-order condition of $rA(r)$ alone does not determine the solution.
There are two candidate solutions to \ref{eq:reporting-problem}.
One is the \textit{safe option}, $r=a+\pi^R$, which guarantees acceptance.
Another is the \textit{risky option}, denoted by $r^+$, which solves the following auxiliary problem:
\begin{equation}\label{eq:RP_risky_def}
    r^{+} \coloneq \argmax_{r\ge 0}\; r\,\big( 1-\ab{F}(r-\pi^R) \big).
\end{equation}
This problem coincides with the reporting problem $\max r A(r)$ on $(a+\pi^R,b+\pi^R]$, but unlike the original problem, it assumes that the acceptance probability $1-\ab{F}(r-\pi^R)$ applies over the entire domain $r\ge 0$.
I define $r^+$ this way because the solution to \ref{eq:reporting-problem} reduces to a comparison of the safe and risky options:
\begin{equation}
    r^* = \max\{a+\pi^R, r^{+}\}.
\end{equation}
The manager chooses the risky option if and only if $r^+ > a+\pi^R$.
This comparison is valid, even though the risky option is rejected with some probability and thus the manager's payoff is $v(r^+) = r^+ A(r^+)<r^+$.
The reason is revealed preference: since $r^+ (1-\ab{F}(r^{+}-\pi^R))\ge (a+\pi^R)(1-\ab{F}(a))=a+\pi^R$ by construction, the risky option must be infeasible (i.e., $r^{+} < a + \pi^R$) whenever the manager chooses the safe option.
See \zcref{fig:acceptance-prob-normal} for an illustration.

\begin{figure}[ht]
    \centering
    \caption{Acceptance Probability}
    \label{fig:acceptance-prob-normal}
    \scalebox{1}{%
        \includegraphics[alt={Acceptance Probability Function}]{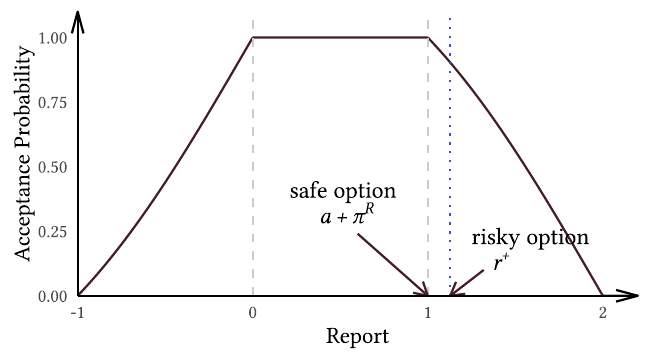}
    }
    \floatfoot{Note: The figure illustrates the acceptance probability $A(r)$ when $X\sim \mathcal{N}(1, 1)$, $\pi^R=1$, and $[a,b]=[0,1]$.
    }
\end{figure}

To obtain the risky option $r^{+}$, rewrite the objective of \eqref{eq:RP_risky_def} in terms of the parent distribution function:
    \begin{align*} 
        r\,\big( 1-\ab{F}(r-\pi^{R}) \big) 
         & =r\,\bigg(1-\frac{F(r-\pi^{R})-F(a)}{F(b)-F(a)}\bigg),\\
         & \propto r\,(F(b)-F(r-\pi^{R})),
    \end{align*}
where the second line follows by multiplying the first line by the positive constant $F(b)-F(a)$.
The first-order condition of \eqref{eq:RP_risky_def} is thus
    \begin{equation}\label{eq:risky-option-FOC}
        F(b) - F(r-\pi^R) = r f(r - \pi^R).
    \end{equation}
Since any $X\in [r-\pi^R,b]$ leads to acceptance, a unit increase in the report yields a marginal benefit of $F(b) - F(r-\pi^R)$.
This is the left-hand side of \eqref{eq:risky-option-FOC}.
The right-hand side is the marginal cost of inflating the report. It reflects the increased rejection probability, which rises only at the lower end of the interval (i.e., when $X=r-\pi^R$).
Near the upper bound $b$, the report is already too high to be accepted ($r+\pi^R>b$), so marginal changes there have no effect on rejection.

Say that the manager \textit{gambles} when he chooses the risky option, $r^*=r^{+}$.
How does the message $[a,b]$ influence the manager's incentive to gamble?
The first consideration is the informativeness of the message---the \textit{length} of the interval.
In determining $r^{+}$, the manager is only concerned with the upper bound $b$ of $X$.
Indeed, $r^+$ determined by \eqref{eq:risky-option-FOC} is independent of the lower bound $a$.
By contrast, the safe option $a + \pi^R$ is pinned down by the lower bound $a$, regardless of the upper bound $b$.
Therefore, for a fixed upper bound $b$, the manager is more likely to gamble when the lower bound $a$ is smaller (i.e., the message is vaguer).
A smaller $a$ makes the safe option less attractive without affecting the risky one.

The second consideration is the content of the message---the \textit{location} of the interval.
Consider messages of fixed length $b - a \equiv \Delta \le 2\pi^R$, and write the message $[a, b]$ as $[a, a+\Delta]$.
The manager's expected payoff from the safe option is $(a+\pi^R)A(a+\pi^R) = a+\pi^R$. 
A marginal increase in the report $r$ from the safe option changes the manager's expected payoff by:
\begin{equation}\label{eq:RP-MBMC}
    (r\,A(r))'|_{\,r=a+\pi^R} = 1 - (a+\pi^R)|A'|,
\end{equation}
where $A'<0$ is the right-derivative of $A$ at $a+\pi^R$.
The first term is the marginal benefit of increasing $r$ by one unit, fixing the acceptance probability at one.
This is always one, \textit{regardless of the location} of the message.
In contrast, the marginal cost, which is the second term of \eqref{eq:RP-MBMC}, increases with $a$.
The report $a + \pi^R$ is now rejected with smaller probability, and the manager loses a fraction $A'$ of the safe option $a + \pi^R$.
\footnote{To be precise, the right derivative $A'$ at $a + \pi^R$ also changes with $a$ unless $X$ is uniformly distributed. Using log-concavity, one can show that $|A'|$ evaluated at $a+\pi^R$ is (weakly) increasing in $a$.}
Consequently, the manager has a stronger incentive to gamble when the message is more negative (i.e., $a$ is smaller).

In sum, vaguer messages make the manager more likely to gamble, as do more negative messages that imply lower realizations of $X$.

\begin{example}[Uniform Distribution]\label{ex:RP-Unif}
    Suppose that $X$ follows the uniform distribution $\mathcal{U}[a,b]$.
    When $b-a\le 2\pi^R$, the solution to \ref{eq:reporting-problem} is given by
    \begin{equation}\label{eq:uniform-RP-sol}
        r^{*}\times A(r^{*})=
        \begin{cases}
            \frac{b+\pi^{R}}{2}\times\frac{b+\pi^{R}}{2(b-a)} & \text{if }a<\frac{b-\pi^{R}}{2},\\
            (a+\pi^{R})\times1 & \text{if }a\ge\frac{b-\pi^{R}}{2}.
        \end{cases}
    \end{equation}
    The manager gambles and chooses the risky option $r^+=(b+\pi^R)/2$ when the message is sufficiently vague (i.e., $a<(b-\pi^R)/2$).
    Moreover, as $b$ becomes smaller, the intervals that prevent gambling shrink: the manager is more likely to gamble when the message contains more negative information (location effect).
\end{example}

\section{Precise Communication and the Gatekeeper's Dilemma}

\subsection{Precise Information Is Exploited}
Suppose that the auditor uses precise communication (i.e., $\mathcal{M}=\Mprec$).
Once the auditor communicates $X$, the manager learns that the auditor will accept any report in the interval $[X - \pi^R, X + \pi^R]$. 
Therefore, the manager reports
\begin{equation*}
    r^* = \max\{X+\pi^R, 0\}.
\end{equation*}
The auditor is always indifferent, receiving a payoff of $-(\pi^R)^2$. The auditor is perfectly gamed.

To avoid being gamed, the auditor can remain silent. 
The auditor's precise communication strategy $\sigma$ induces a \textit{silence set}, $\nd\coloneq\{X\mid \sigma(X)=\varnothing\}$, which is the set of $X$ realizations she does not reveal. 
The auditor's silence signals that $X\in \nd$.

\subsection{When to Stay Silent}
What is an auditor-optimal silence set $\nd$?
One obvious candidate is to always stay silent (i.e., $\nd=[\Xl,\Xh]$).
This is an equilibrium: if the auditor deviates and reveals $X$, then the manager will report $r=X+\pi^R$, which makes the auditor indifferent. 
Can the auditor do better?

By selectively revealing certain pieces of information, the auditor can give meaning to her silence.
The auditor's expected payoff is given by
\begin{equation}\label{eq:optimal-ND-obj}
    \prob(X\in\nd)\E[\max\{u(r_0,X), -(\pi^R)^2\}\mid X\in\nd]+ \prob(X\notin\nd)(-(\pi^{R})^{2}),
\end{equation}
where $r_0\coloneq r(\nd)$ is the manager's response to the auditor's silence.
This expression captures the auditor's fundamental tradeoff---the gatekeeping expert's dilemma.
For a fixed $\nd$, the auditor prefers a larger silence set, since $u(r_0, X) \ge -(\pi^R)^2$. 
As $\nd$ grows, however, the manager's report $r_0$ drifts further from the benchmark $X$, reducing $\E[\max\{u(r_0,X), -(\pi^R)^2\} \mid X \in \nd]$.
In short: the auditor can protect herself from being gamed by expanding $\nd$, but at the cost of receiving reports that diverge more from $X$.

Not all silence sets can be supported in equilibrium.
When $X\notin \nd$, the auditor's expected payoff under the prescribed strategy is $-(\pi^R)^2$.
If she deviates and stays silent, then she obtains the payoff of $u(r_0, X) \ge -(\pi^R)^2$.
The inequality is strict when the induced report $r_0$ is favorable (i.e., $|X-r_0|<\pi^R$).
Conversely, when $X\in \nd$, the auditor does not have an incentive to reveal $X$, as doing so yields the lowest payoff of $-(\pi^R)^2$.
Therefore, the equilibrium constraint is
\begin{equation}\label{eq:equilibrium-constraint}
    |X - r_0 | \ge \pi^R \quad \text{for all } X\notin\nd.
\end{equation}
An optimal silence set, denoted by $\nd^*$, is one that maximizes the auditor's payoff \eqref{eq:optimal-ND-obj} across all possible equilibria under the equilibrium constraint \eqref{eq:equilibrium-constraint}.

When $X\,$ follows a uniform distribution, a particularly simple solution is available: the auditor does not communicate at all.

\begin{proposition}\label{prop:uniform-precise-ND}
    Suppose that $X\sim \mathcal{U}[\Xl,\Xh]$.
    If the auditor must precisely reveal the value of $X$ when she communicates (i.e., $\mathcal{M}=\Mprec$), she is indifferent between any of $\,\{\nd\mid |\nd| \ge 2\pi^R\}$. 
    In particular, the no-communication equilibrium is auditor-optimal: $\nd^*=[\Xl,\Xh]$.
\end{proposition}

To understand this result, consider the introductory example where $X\sim \mathcal{U}[1,9]$.
The auditor can induce a favorable report only on the silence set.
Under no communication ($\nd = [1,9]$), the manager reports $r_0 = 8$.
The auditor receives unacceptable reports for $X\le 7$.
By changing the silence set to $\nd'=[1,7]$, the auditor induces favorable reports for some $X\in[1,7]$.
To implement this strategy, however, she must precisely communicate values of $X$ in $[7,9]$.
By staying silent on $[1,7]$, she effectively shifts the range in which she receives favorable reports from $[7,9]$ to some subset of $[1,7]$.
Since $X$ has a constant density, the expected payoff from favorable reports remains unchanged.

For general log-concave distributions for which the density is not constant, it can be valuable to shift the interval over which the auditor receives a favorable report.
The auditor's optimal silence set becomes a proper subset of the support of $X$.

\begin{example}[The Normal Distribution]\label{ex:normal-precise-ND}
    Suppose that $X\sim \mathscr{N}(1,1)$. 
    Let $\pi^R=1$. 
    If the auditor does not communicate anything ($\nd = (-\infty, \infty)$), the manager reports $r_0 \approx 1.78$.
    The auditor's expected loss is then about $0.62$.
    If the auditor could induce any report independent of the silence set, she would choose $r_{0}=\mu=1$ with $\nd \supset [\mu-\pi^{R},\mu+\pi^{R}]$.
    But in equilibrium, the silence set and the induced report are intertwined.
    An optimal silence set is $\nd^* = (-\infty, b^*)$ with $b^* \approx 2.61$, under which the manager reports $r_0^* \approx 1.61$.
    By eliminating the high realizations of $X$ from the silence set, the auditor shifts the manager's report closer to the ideal point $\mu$.
    Since $r_0^* + \pi ^R = b^*$, the equilibrium constraint \eqref{eq:equilibrium-constraint} is binding.
    Under the optimal strategy, the auditor's expected loss is $0.59$.
    The auditor's expected loss is about $11\%$ smaller under the optimal strategy than under no communication.
    In online supplementary material (\href{https://shunsukematsuno.github.io/papers/supplementary-material_Gatekeeping-Experts-Dilemma.pdf\#nameddest=subsection.C.1}{Appendix C.1}), I characterize an optimal silence set for general log-concave distributions and provide more detail on this example.
\end{example}

This analysis foreshadows the value of vague communication.
The auditor wants to guide the manager's report while maintaining some uncertainty. 
Precise communication forces a stark choice: stay silent, or speak and be gamed.
Vague language offers a way to achieve both of the auditor's goals.

\section{Gatekeeping with Vague Language}

\subsection{The Optimal Communication Strategy with Vague Language}

\subsubsection{Maximal Acceptance Lemma}
A vague communication strategy $\sigma$ induces a partition of the support of $X$.
Represent the induced partition by $\mathcal{D}=\bigcup_{i\in I} D_i$, where $I$ is a countable index set, $D_i\coloneq[d_i,d_{i+1}]$, and $\sigma([d_i,d_{i+1}])=D_i$.
\footnote{I represent the partition as a collection of (closed) intervals, instead of cutoff points. Thus, each element in the partition is an interval corresponding to a message.}
Let $r_i\coloneq r(D_i)$ be the manager's report upon learning that $X\in D_i$. 

An auditor-optimal partition is characterized by the following optimization problem:
\begin{equation}\tag{$\mathrm{OP}$}\label{eq:optimal-partition-problem}
    \min_{\mathcal{D}} \sum_{i\in I} \prob(X\in D_i) \E[\alpha(r_i,X)(r_i - X)^2 + (1-\alpha(r_i,X)) (\pi^R)^2 \mid X\in D_i]. 
\end{equation}
For each $X\in D_i$ the manager reports $r_i$, and thus the auditor's loss is $\min\{(r_i-X)^2, (\pi^R)^2\}$.
In an optimal partition, each interval must be neither too small, lest the manager's report become barely acceptable, nor too large, lest the report become too far from $X$.

The optimal partition problem \eqref{eq:optimal-partition-problem} takes into account the manager's equilibrium behavior through $r_i$ and the auditor's equilibrium acceptance rule through $\alpha(r_i,X)$.
What is missing is the equilibrium condition on the auditor's communication strategy $\sigma$.
For a partition $\mathcal{D}$ to be an equilibrium, the auditor must be deterred from deviating to a different message after observing $X$ (see \zcref{def:PBE}).
For now I set aside this incentive constraint and return to it after solving \ref{eq:optimal-partition-problem}.

The problem \ref{eq:optimal-partition-problem} is potentially complex, as the space of feasible partitions is large.
Yet the following lemma shows that it is without loss to restrict attention to a narrower class of partitions.

\begin{lemma}[Maximal Acceptance Lemma]\label{lem:maximum-acceptance-lemma}
    Consider the vague communication case (i.e., $\mathcal{M}=\Mvag$).
    In an auditor-optimal outcome, the ex-ante probability of acceptance is maximized among all possible equilibria. 

    Moreover, without loss of generality, an auditor-optimal partition can be restricted to those satisfying the following properties:
    \begin{itemize}\setlength{\itemsep}{1pt}\setlength{\parskip}{0pt}
        \item For each interval in the partition, the manager's report is accepted with probability either zero or one.
        \item At most one interval in the partition has the resulting acceptance probability of zero.
    \end{itemize}
\end{lemma}

The result is of independent interest, beyond its role in solving the optimal partition problem.
It shows that the auditor leaves no room to improve the likelihood of receiving an acceptable report.
To see why, consider a communication strategy that involves a message $D = [a, b]$ with $b - a > 2\pi^R$.
In response to $D$, the manager reports $r$, which falls in the upper part of $[a, b]$. 
Then the auditor with realization $X\in [a,r-\pi^R]$ finds the report $r$ unacceptable and rejects it. 
Now consider another communication strategy that differs only on $D$ by splitting $D$ into $[a,r-\pi^R]$ and $[r-\pi^R, b]$.
In the lower subinterval, the manager proposes a report that is acceptable for some $X\in[a,r-\pi^R]$.
Thus the auditor's payoff improves conditional on $X\in[a,r-\pi^R]$.
When $X\in [r-\pi^R,b]$, the manager's report may rise above $r$, potentially reducing the auditor's payoff there.
As the lemma's proof shows, however, this negative ``externality'' is negligible, because the manager's incentive to inflate the report is not sensitive to the lower end of the interval.

In an auditor-optimal outcome, some message may lead to an unacceptable report, because the manager's default report is constant at $r=0$.
If the realized $X$ might fall well below zero, the manager would report $r=0$ regardless of how the auditor communicates.
\zcref{lem:maximum-acceptance-lemma} says that the auditor can bundle all such information into one interval with zero acceptance probability.

\subsubsection{No-Gambling Constraint}
Owing to the Maximal Acceptance lemma, our task is now to identify messages that lead to acceptable reports.
Thus I turn back to the manager's reporting problem.
I analyze when the manager chooses the safe option over the risky one.

In order to ensure that there is a report that is acceptable with probability one, the interval length must be no greater than $2\pi^R$.
This is not sufficient to ensure no gambling.
In \zcref{ex:RP-Unif}, we have seen that the message $D=[a,b]$ must satisfy $a \ge (b -\pi^R)/2$ in addition to $b-a\le 2\pi^R$.
To discourage the manager from selecting the risky option, the auditor cannot be too vague: for each right endpoint $b$, the left endpoint $a$ must be sufficiently large.

Let $\Gamma_G(b) < b$ denote the minimum value of $a$ for which the manager, conditional on receiving message $D=[a,b]$, selects the safe option.
\footnote{\zcref{lem:reporting-acceptance} shows that such a value exists.}
For each right endpoint $b > -\pi^R$, the left endpoint $a$ must satisfy 
\begin{equation}\label{eq:accept-wpone-constraint}
    a \ge \Gamma(b)\coloneq \max \{\Gamma_G(b), b - 2 \pi^R\}.
\end{equation}
I refer to the condition $a \ge \Gamma_G(b)$ as the \textit{no-gambling constraint}.
In the case of a uniform distribution (\zcref{ex:RP-Unif}), this constraint simplifies to $\Gamma_G(b) = (b -\pi^R)/2$.
I call the combined full constraint \eqref{eq:accept-wpone-constraint} the \textit{acceptance constraint}.

When $\Gamma_G(b)\ge b-2\pi^R$, I say that the no-gambling constraint is \textit{relevant}.
In this case, the auditor must provide sufficiently precise information to deter the manager from gambling.
In contrast, when the no-gambling constraint is irrelevant, the auditor needs only to ensure that the interval length is no greater than $2\pi^R$, irrespective of the location of the interval.
The following lemma characterizes when the no-gambling constraint is relevant.

\begin{lemma}\label{lem:no-gambling-const}
    Consider a general distribution of $X$.
    For each right endpoint $b>-\pi^R$, the no-gambling constraint is relevant if and only if $b$ is sufficiently low.
    Specifically, there exists a unique $\hat{b}\in (-\pi^R, \Xh]$ such that $\Gamma_G(b) \ge b-2\pi^R$ if and only if $b \le \hat{b}$.
\end{lemma}

According to this lemma, for high realizations of $X$ (i.e., for $b$ high enough), the no-gambling constraint is irrelevant.
This is because of the location effect discussed in \zcref{sec:RP-General}.
The risky option becomes more attractive as the message shifts rightward at a rate slower than the safe option does, because the risky option is rejected with some probability.

The takeaway from \zcref{lem:no-gambling-const} is that the auditor can afford to be vaguer on the right tail of the distribution of $X$.
This ``vaguer on the right'' principle plays a key role in equilibrium communication strategies.

\subsubsection{The Optimal Partition}
From the Maximal Acceptance property and the acceptance constraint, I can now reformulate the auditor's optimal partition problem:
\begin{equation}\tag{$\mathrm{OP}'$}\label{eq:optimal-partition-prob-relaxed}
    \begin{aligned}
        \min_{\mathcal{D}} \quad & \sum_{i\in I} \prob(X\in D_i) \E[(r_i - X)^2 \mid X\in D_i], \\
        \mathrm{s.t.} \quad & d_i \ge \Gamma(d_{i+1})=\max\{\Gamma_G(d_{i+1}),\,d_{i+1}-2\pi^R\}, \forall i\in I.
    \end{aligned}
\end{equation}
The auditor seeks to minimize the expected loss from accepting the manager's report. 
The constraints ensure that the manager does not select the risky option over the safe option. 
The precise form of $\Gamma$ depends on the distribution of $X$.

\begin{theorem}\label{thm:optimal-partition}
    Let $\mathcal{D}^*$ be a solution to \ref{eq:optimal-partition-prob-relaxed}.
    A communication strategy inducing the partition $\mathcal{D}^*$ constitutes a vague-communication equilibrium.
\end{theorem}

This theorem establishes that the auditor-optimal communication strategy is indeed part of an equilibrium.
In this equilibrium, the manager responds to an off-path message based on his most favorable interpretation of the message (see the appendix for details).
This in turn deters the auditor from deviating to a different message after observing $X$.

The auditor-optimal equilibrium highlights the strategic power of vague communication. 
Under this equilibrium, the auditor carefully balances the tension between being overly transparent and excessively vague. 
To clarify this point, I now explicitly construct the optimal partition.

\subsection{Constructing the Optimal Partition}

\subsubsection{A Direct Approach: The Uniform Case}
Instead of solving \ref{eq:optimal-partition-prob-relaxed} in full generality, I begin with the case of a uniform distribution.
The analytical tractability allows us to discern the key economic forces at play.
For simplicity, I assume that $\Xl \ge -\pi^R$, which ensures that the acceptance probability is one in the auditor-optimal partition (\zcref{cor:accept-wpone-cdn}).

I start with rewriting the optimal partition problem \ref{eq:optimal-partition-prob-relaxed} using the properties of the uniform distribution.
From \eqref{eq:uniform-RP-sol}, the acceptance constraint is $\Gamma(b)=\max\{ b-2\pi^R, (b-\pi^R)/2 \}$.
The expected loss in each interval is $\E[(r_i - X)^2 \mid X\in D_i]=\int_{d_i}^{d_{i+1}} (r_i - x)^2 /(d_{i+1}-d_i)\, dx$.
After I compute this integral, the problem \ref{eq:optimal-partition-prob-relaxed} reads as follows:
\begin{equation}\label{eq:optimal-partition-prob-uniform}
    \begin{aligned}
        \min_{\mathcal{D}} \quad & \sum \frac{d_{i+1}-d_{i}}{\overline{X}-\underline{X}}\left[(\pi^{R})^{2}-\pi^{R}(d_{i+1}-d_{i})+\frac{1}{3}(d_{i+1}-d_{i})^{2}\right], \\
        \mathrm{s.t.} \quad & d_i \ge \max\{ d_{i+1}-2\pi^R, (d_{i+1}-\pi^R)/2 \},\forall i \in I.
    \end{aligned}
\end{equation}

The auditor addresses the gatekeeping expert's dilemma by choosing the appropriate amount of information she provides. 
The constraints $d_i \ge \max\{ d_{i+1}-2\pi^R, (d_{i+1}-\pi^R)/2 \}$ limit how vague the auditor can be without inducing the manager to gamble.
Within this bound, the auditor does not want to reveal too much, but she also wants to avoid being too vague.
As the auditor becomes vaguer (i.e., as $d_{i+1} - d_i$ grows), the expected loss initially falls, reflecting the value of vagueness in deterring gaming.
Beyond a certain point, the expected loss starts to rise, as excessive vagueness induces reports too far from $X$.

To solve \eqref{eq:optimal-partition-prob-uniform}, first suppose the no-gambling constraint is not relevant.
In this relaxed problem, the objective in \eqref{eq:optimal-partition-prob-uniform} and the constraint $d_{i+1}-d_i\le 2\pi ^R$ depend only on the interval length.
Hence, a uniform partition---one that divides the support into equal intervals---achieves the optimal solution.
Denote by $\Delta\coloneq d_{i+1}-d_i$ the interval length.
Represent the conditional expected loss as a function of $\Delta$ by $\ell(\Delta)\coloneq(\pi^R)^2-\pi^R \Delta + \Delta^2/3$.
The loss function $\ell$ is minimized at
\begin{equation}
    \Delta^\mathrm{ideal} \coloneq 1.5\pi^R,
\end{equation}
which satisfies the constraint $\Delta \le 2\pi ^R$.
Thus, ideally, the auditor would like to choose the uniform partition with interval length $\Delta^\mathrm{ideal}$.

However, $\Delta^\mathrm{ideal}$ may not be a feasible solution: intervals of length $\Delta^\mathrm{ideal}$ may not exactly cover the support of $X$.
To fit the support, the auditor must choose the uniform partition whose interval length $\Delta$ is closest to $\Delta^\mathrm{ideal}$.
The equilibrium interval length is therefore given by
\begin{equation}\label{eq:uniform-uniform-partition-problem}
    \Delta^{*}=\frac{\Xh-\Xl}{N^{*}}, \quad
    N^{*}=\argmin_{N\in\mathbb{N}} \left\{\;\Big| \frac{\Xh-\Xl}{N} - 1.5\pi^R \Big| \;\;\middle| \;\; \frac{\Xh-\Xl}{N}\le 2\pi^R\; \right\}.
\end{equation}

The solution involves a simple integer optimization problem.
For a given set of parameters, it is straightforward to solve the problem as illustrated in the following example.
\footnote{A fully explicit solution is provided in the online supplementary material (\href{https://shunsukematsuno.github.io/papers/supplementary-material_Gatekeeping-Experts-Dilemma.pdf\#nameddest=subsection.C.2}{Appendix C.2}).}

\begin{example}[Uniform Partition ($X\sim\mathcal{U}{[1,9]}$ and $\pi^R=1$)]
    \label{ex:uniform-uniform-partition}
    Consider the introductory example discussed in \zcref{sec:example}.
    Observe that $\Xh-\Xl=8$, $8/5=1.6,\text{ and }8/6=1.\dot{3}$.
    Thus the optimal uniform partition has size $N^* = 5$, with interval length $\Delta^* = 1.6$.
    The associated loss is $\ell(1.6)=19/75=0.25\dot{3}$, which is an improvement of about $70\%$ over the no-communication loss of $5/6$.
    \zcref{fig:partition-uniform-example} (right panel) shows the optimal partition. 
    The uniform partitioning satisfies the no-gambling constraint.
    At the lowest interval, this can be verified from $1 \ge (2.6 - 1)/2$.
    Since the no-gambling constraint becomes increasingly slack for higher intervals, it holds throughout the partition.
    Hence the uniform partition is indeed optimal.
\end{example}

So far we have ignored the no-gambling constraint, $d_i \ge (d_{i+1} - \pi^R)/2$.
In the example above, this was without loss.
When is this justified?
From \zcref{lem:no-gambling-const}, the no-gambling constraint is irrelevant when the lower bound of the support $\Xl$ is sufficiently high.
Under a uniform partition, the no-gambling constraint $d_i \ge (d_{i+1}-\pi^R)/2$ simplifies to $d_i+\pi^{R}\ge\Delta$. 
If $\Xl \ge \pi^R$, then $\Delta\le 2\pi^R$ alone guarantees that $d_i+\pi^R-\Delta \ge \Xl+\pi^R - 2\pi^R\ge 0$ for all $d_i\in [\Xl,\Xh)$.
Thus $\Xl \ge \pi^R$ suffices to ensure the optimality of the uniform partition.

When $\Xl < \pi^R$, the no-gambling constraint may bind, and the optimal partition may no longer be uniform.
When the realizations of $X$ are small, the manager's incentive to inflate the report is strong, so the auditor must provide more information to ensure an acceptable report.

In this scenario, the problem \eqref{eq:optimal-partition-prob-uniform} becomes an infinite-dimensional, nonlinear optimization problem. 
Unlike the uniform-partition case, where the problem reduces to selecting an optimal number and length of intervals, we must now consider all possible partitions of the support of $X$.
Explicit analytical solutions become generally infeasible. 
Nevertheless, under the uniform distribution, it is straightforward to solve the problem by computing the optimal value for each size $N$ of the partition.
The following example illustrates the solution when the no-gambling constraint binds.

\begin{example}[Non-Uniform Partition]
    Suppose that $X \sim \mathcal{U}[0,6]$ and $\pi^R=1$.
    The sufficient condition for a uniform partition, $\Xl \ge \pi^R$, is violated.
    Since $(6-0)/4 = 1.5=\Delta^{\mathrm{ideal}}$, the optimal uniform partition has size $N^* = 4$.
    In the leftmost interval $[0, 1.5]$, the no-gambling constraint is violated: from \eqref{eq:uniform-RP-sol}, the manager's report is $r=1.25$, while the safe option is $r=1$.
    To construct an auditor-optimal equilibrium, suppose that the no-gambling constraint is relevant only in the first interval $[0, d_1]$.
    Over the remaining interval $[d_1, 6]$, a uniform partition is optimal.
    The auditor-optimal equilibrium is then obtained by choosing $d_1$ to minimize expected loss, subject to the relevant constraints.
    The optimal cutoff is $d_1^*=1$ and yields the partition $\mathcal{D}^*=[0,1]\cup[1,1+\Delta]\cup[1+\Delta,1+2\Delta]\cup[1+2\Delta,6]$ with $\Delta=5/3\approx 1.67$. 
    Under this partition, the auditor's loss is $22/81\approx 0.27$. 
    By similar calculations, one can derive the optimal loss when the no-gambling constraint is relevant for two or more intervals and verify that the above solution is indeed optimal.
\end{example}

\subsubsection{A Dynamic Programming Approach}
Since the reporting problem \ref{eq:reporting-problem} lacks a closed-form solution outside the uniform-distribution case, it is challenging to explicitly solve for the auditor's optimal communication strategy.
Thus I adopt a dynamic programming approach to derive the auditor's strategy.
\footnote{\cite{aliVoluntaryDisclosurePersonalized2023} use a similar approach to derive the optimal segmentation of buyer types in a monopolistic screening setting.}
This method is more flexible and broadly applicable but offers limited economic intuition.
Accordingly, I relegate the details to the appendix.

Denote by $L(b)$ the auditor's expected loss in an auditor-optimal equilibrium when the original distribution of $X$ is truncated to $[\Xl, b]$.
Define $\ell(a,b)\coloneq \E[(r([a,b]) - X)^2 \mid X\in [a,b]]$ as the expected loss from accepting the manager's report when $X\in [a,b]$.

\begin{proposition}\label{prop:partition-DP}
    The optimal partition of $[\Xl,b]$ is characterized by 
    \begin{equation}\label{eq:DP-problem}
        L(b) = \min_{a \in [\Gamma(b), b]} \prob(X \ge a \mid X\le b) \ell(a,b) + \prob(X \le a \mid X \le b) L(a). 
    \end{equation}
    This functional equation admits a unique solution $L$.
\end{proposition}

Let $\lambda(b)$ be the solution to \eqref{eq:DP-problem}. 
Then $[\lambda(b), b]$ forms the rightmost interval in the optimal partition of $[\Xl, b]$.
Starting from $\lambda(b)$, we solve \eqref{eq:DP-problem} again to identify the next interval, $[\lambda(\lambda(b)), \lambda(b)]$, and so on.
The sequence $\{\lambda^{(i)}(b)\}$ defined recursively by $\lambda^{(i+1)}(b) = \lambda(\lambda^{(i)}(b))$ gives the cutoffs that define the auditor-optimal partition of $[\Xl, b]$.
The limit $\lim_{b \to \Xh} L(b)$ is the auditor's expected loss in the original problem \ref{eq:optimal-partition-prob-relaxed}.

\section{Comparative Statics}\label{sec:compara-stat}
I now study how the auditor's optimal communication strategy changes with the auditor's independence ($\pi^R$) and the transaction characteristics ($\theta\,$).
In general, formal comparative statics are difficult, as the optimal partition problem \eqref{eq:optimal-partition-prob-relaxed} does not admit a closed-form solution.
An important exception is the case of uniform distributions.
Outside this special case, I rely on numerical simulations to illustrate the comparative statics.

\subsection{The Effect of Independence}
How does the independence parameter $\pi^R$ affect the auditor's optimal communication strategy?
Recall that a lower $\pi^R$ corresponds to a lower cost of rejection and thus a more independent auditor.
I analyze the equilibrium \textit{amount of information} the auditor provides under vague communication.
For a partitional communication strategy $\sigma$, define the amount of information by the overall reduction in the variance of $X$:
\begin{equation}\label{eq:info-amount}
    \info(\sigma)\coloneq \frac{\mathrm{Var}(X) - \E[\mathrm{Var}(X\mid \sigma(X))]}{\var(X)}.
\end{equation}
When the auditor communicates more precisely, each of $\mathrm{Var}(X\mid \sigma(X))$ becomes smaller and thus the amount of information is larger.
I normalize the amount of information by the prior variance, so that it ranges from $0$ (no information) to $1$ (full information).

The uniform distribution case admits a sharp characterization:
\begin{proposition}\label{prop:compara-pi-uniform}
    Suppose that $X$ follows a uniform distribution with $\Xl>0$. 
    The auditor communicates more precisely as she becomes more independent for $\pi^R \in (0, \Xl)$:
    (i)  the number of intervals is weakly decreasing in $\pi^R$, and 
    (ii) the amount of information, $\info(\sigma)$, is weakly decreasing in $\pi^R$.
\end{proposition}

The result is intuitive.
A more independent auditor is more strict in rejecting the manager's report.
Since the optimal communication strategy must ensure that the manager's report is always acceptable, the auditor must partition the support of $X$ into smaller intervals.
Monotone comparative statics formalize this intuition \citep{milgromMonotoneComparativeStatics1994}.
The condition $\pi^R\in(0,\Xl)$ guarantees that a uniform partition is optimal.

The second part is a direct consequence of the first part.
The posterior variance is the same for each possible message $D=[d_i,d_{i+1}]$ and is the variance of the uniform distribution over $D$:
\begin{equation}
    \mathrm{Var}(X\mid D_i) = \frac{(d_{i+1}-d_i)^2}{12} = \frac{(\Delta^{*})^2}{12}.
\end{equation}
If the number of intervals is weakly decreasing in $\pi^R$ (part (i)), then the interval size $\Delta^{*}$ is weakly increasing in $\pi^R$.
Therefore, $\mathrm{Var}(X\mid D_i)$ is weakly increasing in $\pi^R$.

What if a uniform partition is not optimal?
In that case, the no-gambling constraint $d_i\ge (d_{i+1}-\pi^R)/2$ becomes crucial.
The auditor must depart from a uniform partition to ensure that the manager's report is always acceptable.
To illustrate, consider a small increase in $\pi^R$.
The no-gambling constraint becomes easier to satisfy, as the auditor becomes willing to accept a wider range of reports.
Let $\mathcal{D}^\mathrm{ideal}$ be the solution to the relaxed version of the optimal partition problem \ref{eq:optimal-partition-prob-relaxed}, in which the no-gambling constraint is ignored.
Since a uniform partition is optimal under the relaxed problem, denote by $N^\mathrm{ideal}$ the number of intervals in $\mathcal{D}^\mathrm{ideal}$.
When $\pi^R$ rises but $N^\mathrm{ideal}$ does not change, the actual partition $\mathcal{D}^*$ moves toward the ideal uniform partition $\mathcal{D}^\mathrm{ideal}$.
Consequently, the posterior variance decreases (more information).

\zcref{fig:uniform-nonuniform-compara} illustrates how the independence parameter $\pi^R$ affects the equilibrium communication strategy.
The left column shows the auditor's optimal partitioning (top) and the amount of information as $\pi^R$ varies (bottom).
Consistent with \zcref{prop:compara-pi-uniform}, the partition becomes coarser as the auditor becomes less independent, and the amount of information correspondingly decreases.
When a uniform partition is not feasible, the right column highlights a non-monotonic relationship between $\pi^R$ and the amount of information.
When the auditor becomes less independent, if the number of intervals stays the same (e.g., $\pi^R = 2.3$ to $\pi^R=2.5$), then the partition approaches the ideal uniform partition. 
As a result, the amount of information increases in response to an increase in $\pi^R$.

The mechanics described above extend beyond the uniform distribution cases.
The distortion from the relaxed partition $\mathcal{D}^\mathrm{ideal}$ occurs when the no-gambling constraint is relevant.
Consider a small decrease in $\pi^R$ (i.e., the auditor becomes more independent).
The auditor is willing to accept a narrower range of reports, so she must be more precise.
This effect is stronger when $X$ is low than when $X$ is high, because the no-gambling constraint is relevant only for small realizations of $X$ (\zcref{lem:no-gambling-const}).
As a result, the auditor can now afford to be vaguer when $X$ is high, because there is more room to adjust the communication strategy to be closer to $\mathcal{D}^\mathrm{ideal}$.
In other words, as $\pi^R$ decreases, the distortion from $\mathcal{D}^\mathrm{ideal}$ may shrink for the right tail of the distribution of $X$.
When the ``vaguer on the right'' effect dominates, the amount of information decreases as the auditor becomes more independent.

The preceding discussion suggests that the negative effect of increased independence on the auditor's information provision is more likely when the distribution of $X$ places greater mass near the origin, where the no-gambling constraint is relevant.
The top panels of \zcref{fig:normal-compara} illustrate this point for the case of a normal distribution.
The left panel plots the amount of information against $\pi^R$ when the mean of $X$ is low; the right panel does so when the mean is high.
In the low-mean case, as the auditor becomes less independent (i.e., as $\pi^R$ increases), the amount of information first rises, then declines.
In the high-mean case, the region of $\pi^R$ in which information increases as the auditor becomes less independent virtually disappears.

\begin{figure}[ht]
    \centering
    \caption{Comparative Statics: Uniform vs. Non-Uniform Partition}
    \label{fig:uniform-nonuniform-compara}
    \scalebox{.9}{%
        \includegraphics[alt={Comparative Statics Illustration for Uniform Distribution}]{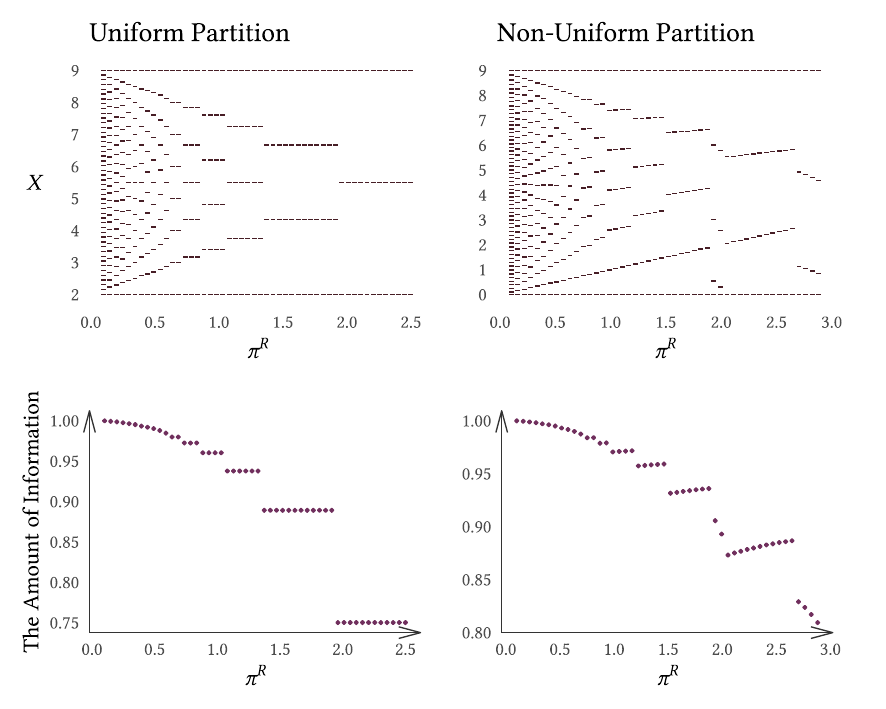}
    }
    \floatfoot{Note: The figure compares the auditor's optimal partitioning and the amount of information, defined in \eqref{eq:info-amount}, under uniform distributions. The left column corresponds to the case of a uniform partition ($X\sim\mathcal{U}[2,9]$), the right column to one with a non-uniform partition ($X\sim\mathcal{U}[0,9]$).
    The top panels show the partition structure as a function of the independence parameter $\pi^R$; each horizontal segment represents an interval. The bottom panels plot the amount of information against $\pi^R$.
    The amount of information is normalized by the prior variance, so that it ranges from $0$ to $1$.
    }
\end{figure}

\subsection{The Effect of Transaction Complexity}
So far I have fixed the transaction characteristics $\theta$ (i.e., the distribution of $X$).
Now I analyze how $\theta$ affects the equilibrium communication.
I vary the variance of $X$ while holding its mean fixed.
A higher variance of $X$ corresponds to greater transaction complexity and thus greater importance of the auditor expertise.
More complex transactions, such as mergers and acquisitions, require greater judgment and estimation to apply relevant standards.
Consequently, the auditor's preferred report has more uncertainty ex ante, and the auditor has more to guide the manager on.

Specifically, I consider the family of distributions $\{F(\cdot\,;\theta)\}_{\theta \in \Theta},\, \Theta \coloneq[0,\infty)$ with differing variances:
\begin{equation}
    \mathrm{Var}(X\mid \theta) \text{ is increasing in } \theta.
\end{equation}
Higher values of $\theta$ correspond to more complex transactions and, consequently, to greater information asymmetry between the auditor and the manager.
For concreteness, I consider a normal distribution case: $X\sim \mathcal{N}(\mu,\sigma^2_\theta)$, where $\sigma^2_\theta\coloneq \sigma_0^2 (1+ \theta)$ and $\sigma_0^2>0$ is the constant baseline variance.

There are several forces at play.
\footnote{The logic described applies in general to any strictly log-concave distribution, which is unimodal.}
Suppose that the mean of $X$ is low, say $\mu<0$.
Fix the equilibrium communication strategy and consider a marginal increase in $\sigma_\theta$. 
Recall that the no-gambling constraint is relevant only when $X$ lies near the origin, where the auditor must be relatively precise.
When the mean of $X$ is low, increasing the variance shifts probability mass away from the origin and into the right tail, where the no-gambling constraint is irrelevant.
As a result, under the fixed communication strategy, the new distribution places more weight on regions where the auditor is less precise. 
The amount of information thus decreases relative to the increase in $\sigma_\theta$.
I call this the \textit{statistical effect}.

Second, consider how the equilibrium communication strategy itself changes with $\sigma_\theta$.
The key is how transaction complexity alters the manager's incentive.
If $[a,b]$ lies in the right tail, an increase in the variance of $X$ pushes more probability mass to the right.
This makes the manager more likely to inflate his report, prompting the auditor to provide more information. 
When $\mu$ is low, a majority of messages fall to the right of the mean.
Consequently, increasing $\sigma_\theta$ encourages gambling by the manager, and the auditor must provide more information on average.
I call this the \textit{strategic effect}.

In summary, the two forces can work in opposite directions.
The statistical effect is the mechanical impact of altering the prior distribution; holding the communication strategy fixed, a change in $\sigma_\theta$ shifts the probability mass between regions where the no-gambling constraint is relevant and where it is not.
When this effect dominates, the amount of information declines as $\sigma_\theta$ rises (relative to the increase in $\sigma_\theta$). 
The strategic effect arises from the manager's incentive to inflate his report; a change in $\sigma_\theta$ alters the manager's incentive, which in turn induces the auditor to adjust her communication strategy.
When this adjustment in the communication strategy dominates, the amount of information may increase.

The bottom panels of \zcref{fig:normal-compara} illustrate the comparative statics.
When $\mu$ is low, most realizations of $X$ lie in the right tail.
Thus the statistical effect can be strong enough, and we see this in the left bottom panel of the figure. 
When $\mu$ is high, we see in the right bottom panel that the strategic effect tends to dominate, especially when $\sigma_\theta$ is low.

\begin{figure}[ht]
    \centering
    \caption{Comparative Statics: Normal Distribution}
    \label{fig:normal-compara}
    \scalebox{.85}{%
        \includegraphics[alt={Comparative Statics Illustration for Normal Distribution}]{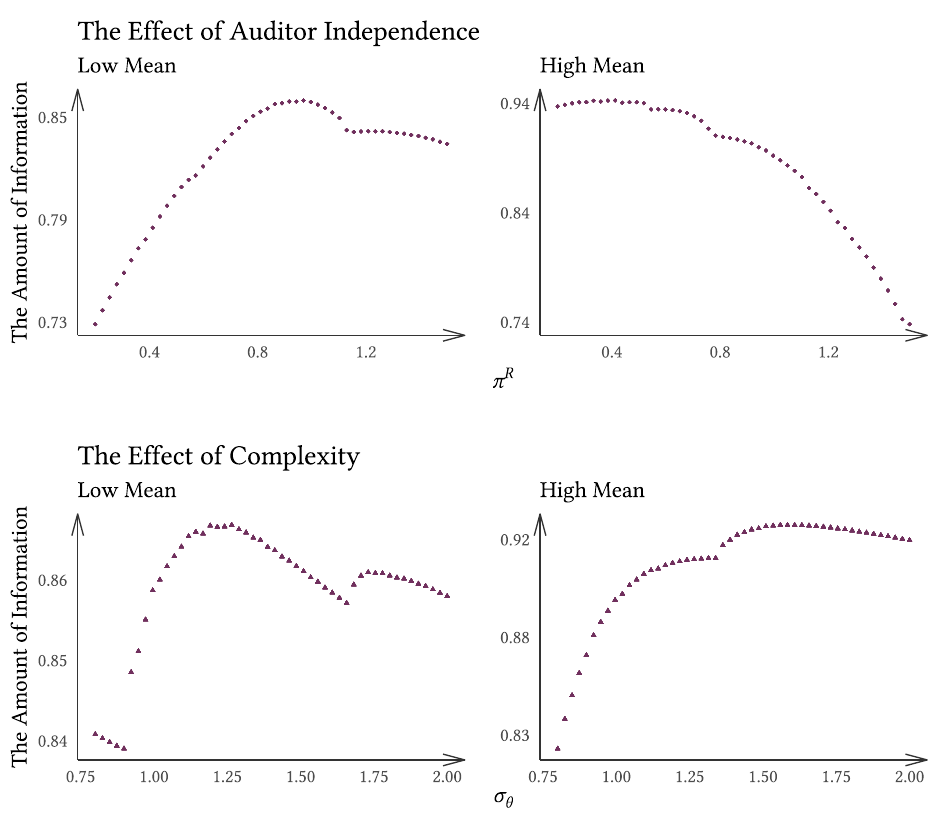}
    }
    \floatfoot{Note: The figure illustrates the comparative statics for normal distributions, $\mathcal{N}(\mu, \sigma^2_\theta)$. 
    The top panels plot the amount of information against $\pi^R$; the bottom panels against $\sigma_\theta$.
    The amount of information is normalized by the prior variance, so that it ranges from $0$ to $1$.
    In the top panels, the low-mean case is $\mathcal{N}(-0.5, 1)$, and the high-mean case is $\mathcal{N}(0.5, 1)$.
    In the bottom panels, the low and high means are again $-0.5$ and $0.5$, respectively, with $\pi^R=1$ fixed.
    }
\end{figure}

\subsection{Empirical Implications}
The analysis yields several empirical implications for financial auditing.
The model predicts a non-monotonic relationship between auditor independence and the amount of information the auditor conveys.
It also implies a nuanced link between transaction complexity and the auditor's communication strategy.
Normally researchers do not observe auditor--client communications. 
However, insofar as the auditor's expertise improves the accuracy of financial statements, financial reporting quality measures---such as restatements---may serve as proxies for the depth of auditor--manager communication.
Auditor independence could be proxied by high non-audit fees or extended auditor tenure \citep{defondReviewArchivalAuditing2014}.
\cite{hoitashMeasuringAccountingReporting2018} propose a measure of accounting reporting complexity based on 10-K filings.
\footnote{See also \cite{chychylaComplexityFinancialReporting2019} for a study examining how reporting complexity relates to reporting outcomes and auditor expertise.}
In addition, principles-based versus rules-based accounting standards could proxy for the degree of information asymmetry between auditor and manager \citep{caplanModelAuditingBrightLine2004,folsomPrinciplesBasedStandardsEarnings2017}.
Principles-based standards require more judgment and estimation and thus create greater information asymmetry.
Examining the relationship between reporting quality and these proxies for transaction complexity could shed new light on how auditor--manager communication shapes financial reporting.

\section{Concluding Remarks}\label{sec:conclusion}
There are many important aspects of gatekeeping that are left out of the model. 
The current paper's stylized model provides a foundation for developing a more comprehensive theory of gatekeeping experts.
I discuss several directions for future research. 

\subsection{Toward a Comprehensive Theory of Gatekeeping Experts}
\subsubsection{Endogenous Choice of Transaction: Efficiency--Gameability Tradeoff}
Potentially the most important avenue for future research is to endogenize the choice of transaction.
In the model, the transaction is fixed, and the auditor's role is to discipline the manager's report.
In practice, the manager has the ability to choose the transaction, taking into account its reporting implications \citep{roychowdhuryEarningsManagementReal2006}.
With such an endogenous transaction choice, we can study the real effects of gatekeeping experts on the economy.
Here I sketch a simple extension of the model to illustrate this point.

Suppose that the manager can choose from a set of transactions $\theta\in \Theta$.
\footnote{Here there is no restriction on the set of transactions $\Theta$.}
Each $\theta$ corresponds to a distribution $F(\cdot\,;\theta)$ of $X$ over $[\Xl_{\theta}, \Xh_{\theta}]$.
The manager chooses $\theta$ before engaging with the auditor. 
For simplicity, suppose that $\Xl_\theta>-\pi^R$ for all $\theta\in \Theta$, so that acceptance occurs with probability one (see \zcref{cor:accept-wpone-cdn}).
Let $\E_\theta[X]$ denote the expected value of $X$ under transaction $\theta$.
I assume that $\E_\theta[X]$ represents the economic value of the transaction.
Thus an \textit{efficient} transaction maximizes $\E_\theta[X]$ among $\theta \in \Theta$.

In general, the manager does not choose an efficient transaction.
Recall that each $X$ corresponds to a unique message interval $[a_{\theta}(X), b_{\theta}(X)]$.
Under a communication strategy that satisfies the maximal acceptance, the manager reports $r_{\theta} = a_\theta(X) + \pi^R$, which is accepted for sure.
Thus the manager chooses $\theta$ to maximize $\E_\theta[r_\theta]$.
The key identity is 
\begin{equation}\label{eq:gameability}
    \E_\theta[r_\theta] = \E_\theta[X] + k(\theta),\quad k(\theta)\coloneqq \E_{\theta}[a_{\theta}(X)-X]+\pi^{R}.
\end{equation}
The term $k(\theta)$ captures the \textit{gameability} of transaction $\theta$.
This is the expected amount the manager can inflate the report above $X$.
The gameability becomes smaller when the auditor is more independent (lower $\pi^R$).
From \eqref{eq:gameability}, the manager chooses $\theta$ to maximize $\E_\theta[X] + k(\theta)$. 
Therefore we have the following observation.
\begin{observation*}\label{obs:efficiency-gameability}
    For a generic set of transactions $\Theta$, the manager chooses a transaction that does not maximize the economic value $\E_\theta[X]$.
\end{observation*}
Since $\E_\theta[X]$ and $k(\theta)$ are generally not aligned, there is a tradeoff between efficiency and gameability.
\begin{example}[Efficiency--Gameability Tradeoff]
    Suppose that $\pi^R=20$ and the manager chooses between two transactions.
    Let $\Theta=\{L,H\}$ with $X_L\sim\mathcal{U}[21,63]$ and $X_H\sim\mathcal{U}[30,66]$.
    Transaction $H$ first-order stochastically dominates $L$ and has a higher economic value, $\E_H[X]=48>42=\E_L[X]$.
    Under $L$, the auditor-optimal vague communication is partitioning the support into $[21,42)$ and $[42,63]$.
    The manager's expected report is $\frac{1}{2}(41+62)=51.5$.
    Under $H$, the auditor uses a single message $[30,66]$.
    The manager's expected report is $50$.
    Thus the manager chooses the inefficient transaction $L$.
    Indeed, its gain in gameability, $k(L)-k(H)=9.5-2=7.5$, exceeds the loss of $6$ in economic value.
\end{example}

\subsubsection{Other Extensions}
Another avenue for future research is to examine how alternative communication technologies, such as cheap talk or costly deception, affect a gatekeeper's ability to influence agents. 
In the model, I assume that the auditor's communication is truthful. 
In the online supplementary material (\href{https://shunsukematsuno.github.io/papers/supplementary-material_Gatekeeping-Experts-Dilemma.pdf\#nameddest=subsection.D.1}{Appendix D.1}), I provide an exploratory analysis of how cheap-talk communication affects the auditor's ability to influence the manager.
There, I construct equilibria without the truthful-communication assumption and show that the auditor can still use vague language to guide the manager without inviting perfect gaming.
This is possible because the auditor has a single-peaked preference and thus her incentive to mislead the manager is limited.

Yet another avenue is to incorporate richer institutional features of specific gatekeeping settings.
In the baseline model, the manager does not possess any private information about $X$.
In the online supplementary material (\href{https://shunsukematsuno.github.io/papers/supplementary-material_Gatekeeping-Experts-Dilemma.pdf\#nameddest=subsection.D.2}{Appendix D.2}), I relax this assumption.
The key insight is that the manager's private information weakens the auditor's influence.

A number of key institutional features remain to be explored. 
For example, in financial auditing and other settings, there are multiple gatekeepers,  interactions between gatekeepers and agents are long-term and repeated, and there are multiple dimensions of issues on which bargaining occurs.
Such analyses would yield more testable predictions and concrete policy implications.

\subsection{Conclusion}
A gatekeeper is a unique institution characterized by veto power without direct control over the agent they oversee.
Often the gatekeeper has expertise and seeks to guide the agent's decisions.
Yet veto power alone cannot ensure that the agent acts in the gatekeeper's interest.
I develop a theory of a gatekeeping expert's dilemma: she wants to guide the agent with her expertise, but sharing too much knowledge invites gaming.
She resolves the dilemma by speaking vaguely---partitioning her information to guide decisions without revealing too much.
By leaving just enough uncertainty, she deters gaming while still influencing the agent's actions.
The paper thus offers a theory of influence through the strategic use of vagueness.

\clearpage

\begin{singlespacing}
    \bibliographystyle{econ-econometrica-nobysame}
    \bibliography{EvidenceAudit}
\end{singlespacing}

\newpage
\appendix
\singlespacing

\part*{\LARGE {Appendix}}
\addcontentsline{toc}{part}{Appendix}  % Add "Appendix" to ToC as a part

\section{Proofs}\label{appsec:Proofs}

This section provides the proofs of the results described in the main text.

\subsection*{Proof of \zcref{lem:reporting-problem-solution}}
If $b\le -\pi^R$, then $A(r)=0$ for all $r\ge 0$, and the unique solution to \ref{eq:reporting-problem} is $r=0$.
Thus, assume that $b>-\pi^R$.
I consider two cases, depending on whether $b-a>2\pi^R$ or not. 

\subsubsection*{Case 1. $b-a> 2\pi^R$}
In this case, the acceptance probability $A(r)$ is given by 
\begin{equation}\label{eq:acceptance-probability-large}
    A(r)
    =\begin{cases}
        \ab{F}(r+\pi^{R}) & \text{if } r\in[a-\pi^{R},a+\pi^{R}),\\
        \ab{F}(r+\pi^R)-\ab{F}(r-\pi^R) & \text{if } r\in[a+\pi^{R},b-\pi^{R}],\\
        1-\ab{F}(r-\pi^{R}) & \text{if } r\in(b-\pi^{R},b+\pi^{R}].
    \end{cases}
\end{equation}
Since $A(r)$ is increasing in $r$ on $[a-\pi^R, a + \pi^R)$, the solution to \ref{eq:reporting-problem} is at least $a + \pi^R$.
On  the region $r\ge a+\pi^R$, the acceptance probability is $A(r)=\ab{F}(\min\{r+\pi^R,b\})-\ab{F}(r-\pi^R)$.
Since the function $A(r)$ has a kink at $r=b-\pi^R$, simply solving the first-order condition of $\max rA(r)$ is not sufficient to characterize the solution.
I thus consider each of the smooth functions $\ab{F}(r+\pi^R)-\ab{F}(r-\pi^R)$ and $\ab{F}(b)-\ab{F}(r-\pi^R)$ separately.
I first prove the result assuming that $f$ is strictly log-concave.
Then I will return to the case in which $f$ is only weakly log-concave.

Consider the following relaxed maximization problems:
\begin{equation}
    \max_{r\ge 0} \,r ~ \left[\frac{F(r+\pi^R) - F(r-\pi^R)}{F(b)-F(a)}\right], \quad 
    \max_{r\ge 0} \,r ~ \left[\frac{F(b) - F(r-\pi^R)}{F(b)-F(a)}\right].
\end{equation}
Note that the problems do not restrict $r$ to be in the corresponding region for each case of \eqref{eq:acceptance-probability-large}.
\footnote{In addition, unlike \eqref{eq:acceptance-probability-large}, the acceptance probability is nonzero for the entire support of $X$, not just on $[a,b]$.}
The first-order condition of each problem reduces to the following fixed-point equations:
\begin{equation}
    r = h_1(r)\coloneq \frac{F(r + \pi^R) - F(r - \pi^R)}{f(r - \pi^R) - f(r + \pi^R)}, \quad
    r = h_2(r)\coloneq \frac{F(b) - F(r - \pi^R)}{f(r - \pi^R)}.
\end{equation}
Routine arguments based on log-concavity show that $h_2$ is monotonically decreasing in $r\le b+\pi^R$ \citep[see, e.g., ][]{bagnoliLogConcaveProbabilityIts2005}.
Thus, $h_2$ admits a unique fixed point $r^+_2 \in (0, b + \pi^R)$ satisfying $r^+_2 = h_2(r^+_2)$.
The function $h_1$ need not be monotone, as $f$ is in general not strictly decreasing everywhere. 
But strict log-concavity guarantees that there is a region where it is strictly decreasing:

\begin{claim}
    The function $r\mapsto f(r-\pi^R)-f(r+\pi^R)$ has a unique root $\hat{r}$.
\end{claim}

\begin{proof}
    Define $w(r)\coloneq g(r-\pi^R) - g(r+\pi^R)$, where $g(r)=\log{f}(r)$. 
    The root of $w$ coincides with that of $ f(r-\pi^R)-f(r+\pi^R)$.
    Strict log-concavity implies that $w'(r) = g'(r-\pi^R) - g'(r+\pi^R)>0$.
    Hence, $w$ has a unique root.  
\end{proof}

Therefore, $h_1(r)>0$ for $r>\hat{r}$, and $\lim_{r\downarrow \hat{r}} h_1(r) = \infty$.  
Strict log-concavity then ensures that $h_1$ is monotonically decreasing for $r\ge \hat{r}$ and admits a unique fixed point $r_1^+ = h_1(r_1^+)$ in that region.

Now we turn to the original problem $\max_r rA(r)$.
Suppose that $r_1^+> r_2^+$ is true.
Then, from the first-order conditions, the problem has a unique solution:
\begin{equation}\label{eq:RP-solution-large}
    r^{*}=
    \begin{cases}
        a+\pi^R & \text{if }r_{1}^{+}<a+\pi^{R},\\
        r_{1}^{+} & \text{if } r_1^+ \in [a+\pi^R, b-\pi^R],\\
        b-\pi^{R} & \text{if } r_{2}^{+}\le b-\pi^{R}<r_{1}^{+},\\
        r_{2}^{+} & \text{if } r_{2}^{+} > b-\pi^{R}.
    \end{cases}
\end{equation}
Thus it suffices to show that $r^+_1 > r_2^+$.

\begin{claim}
    For $r \in (\hat{r}, \infty)$, we have $h_1(r) > h_2(r)$.
\end{claim}

\begin{proof}
    Rewrite the difference as 
    \begin{equation}
        h_1(r)-h_2(r)=\frac{\left[F(b)-F(r-\pi^{R})\right]/f(r-\pi^{R})-\left[F(b)-F(r+\pi^{R})\right]/f(r+\pi^{R})}{(f(r-\pi^{R})-f(r+\pi^{R}))/f(r+\pi^{R})}.
    \end{equation}
    Since $r > \hat{r}$, the denominator of the expression is positive.
    Log-concavity implies that $x\mapsto (F(b) - F(x))/f(x)$ is strictly decreasing for $x\le b$.
    Thus the numerator is positive as well.
\end{proof}

This result implies that $h_1(r^+_1) = r^+_1 > h_2(r^+_1)$.
Since $r-h_2(r)$ is strictly increasing, we conclude that $r_1^+>r_2^+$.

Finally, return to the case in which $f$ is only weakly log-concave. 
The previous analysis of $h_2$ only requires weak log-concavity; thus $\max r(1-\ab{F}(r-\pi^R))$ still admits a unique maximizer $r_2^+$.
The analysis of $h_1$ needs a minor modification.
Let $\overline{r}$ be the largest root of $f(r-\pi^R)-f(r+\pi^R)$.
By the continuity of the mapping, $h_1(r)$ admits a unique fixed point $\overline{r}_1^{+}$ on $(\overline{r},\infty)$.
Define 
\begin{equation}
    r_{1}^{+}\coloneq
    \begin{cases}
      \overline{r}_1^{+} & \text{if }\overline{r}<\infty,\\
      \infty & \text{otherwise}.
    \end{cases}
\end{equation}
The second case applies if $f$ is constant everywhere (i.e., uniform distribution).
The first-order condition implies that \eqref{eq:RP-solution-large} still holds with this modified definition of $r_1^+$.

\subsubsection*{Case 2. $b-a\le 2\pi^R$}
In this case, the acceptance probability is either $1$ or $1-\ab{F}(r-\pi^R)$ for all $r\ge a+\pi^R$.
The relaxed problem $\max_{r\ge 0} r(1-\ab{F}(r-\pi^R))$ is exactly the same as in the previous case.
Thus, the solution to the original problem $\max_r rA(r)$ is 
\begin{equation}
    r^{*}=
    \begin{cases}
       a+\pi^R & \text{if }r_{2}^{+}<a+\pi^{R},\\
       r_{2}^{+} & \text{if }r_{2}^{+}\ge a+\pi^{R},
    \end{cases}
\end{equation}
where $r_2^+$ is defined in the previous case (i.e., the unique fixed point of $r=h_2(r)$).
\qed

\subsection*{Proof of \zcref{prop:uniform-precise-ND}}
An optimal silence set satisfies
\begin{align}\label{eq:optimal-silence-problem}
    \min_{\nd}\, & \prob(X\in\nd)\E[\min\{|r_0-X|,\pi^R\}^2\mid X\in\nd]+\prob(X\notin\nd)(\pi^{R})^{2}\\
    \mathrm{s.t.}\, & |X-r_{0}|\ge\pi^{R}\quad\text{for all }X\notin\nd.
\end{align}
In the online supplementary material, I show that any equilibrium silence set is a connected set (\href{https://shunsukematsuno.github.io/papers/supplementary-material_Gatekeeping-Experts-Dilemma.pdf\#nameddest=lemma.C.1}{Lemma C.1}).
Thus, $\nd=[a,b]$ for some $a<b$, and \eqref{eq:optimal-silence-problem} reduces to a two-dimensional constrained optimization.
In the proof below, I reduce the optimization problem further to that of $b$ by showing that $a=\Xl$ without loss.
Then I show that $b=\Xh$ achieves the optimal value.

\begin{proof}
\step{Step 1. $a=\Xl$ without loss.}
Let $\nd=[a,b]$ for some $a$ and $b$ such that $\Xl\le a<b\le \Xh$.
The explicit solution to \ref{eq:reporting-problem} is given by \eqref{eq:uniform-RP-sol}.
If $b-a\le 2\pi^R$, then we have $r^*+\pi^R>b$, violating the equilibrium constraint \eqref{eq:equilibrium-constraint}.
Thus, $b-a>2\pi^R$ is necessary for an equilibrium.
The solution $r^*$ then does not depend on $a$.
Expanding the silence set only improves the auditor's expected payoff when the report is fixed, so $a=\Xl$ without loss.

\step{Step 2. Optimality of $b=\Xh$.}
Under $a=\Xl$, the auditor's expected loss is given by
\begin{equation}
    \int_{r_0 - \pi^R}^{b}(r_{0}-x)^{2}\frac{1}{\Xh-\Xl}dx+\left[\frac{r_0 - \pi^R -\Xl}{\Xh-\Xl}+\frac{\Xh-b}{\Xh-\Xl}\right](\pi^R)^{2}.
\end{equation}
If $b < 3\pi^R$, then $r^*+\pi^R=(b+3\pi^R)/2>b$, violating the equilibrium constraint \eqref{eq:equilibrium-constraint}.
Thus, only $b\ge \max\{3\pi^R, \Xl + 2\pi^R\}$ are feasible, and $r_0=b-\pi^R$.
Without loss, I can assume that the right-hand side is smaller than $\Xh$.
\footnote{Otherwise $b=\Xh$ is the only feasible equilibrium, and the proof is complete.}
The auditor-optimal equilibrium satisfies
\begin{equation}
    \min_{b \ge  \max\{3\pi^R, \Xl + 2\pi^R\}}\int_{b - 2\pi^R}^{b}(b-\pi^{R}-x)^{2}dx + (\Xh - \Xl - 2\pi^R)(\pi^R)^{2}.
\end{equation}
The first integral term simplifies to $2(\pi^R)^3/3$, so the objective function is constant in $b$.
Therefore, $b=\Xh$ achieves the optimum, and $\nd=[\Xl,\Xh]$ is a solution to \eqref{eq:optimal-silence-problem}.
\end{proof}

\subsection*{Proof of \zcref{lem:maximum-acceptance-lemma}}

Before proving \zcref{lem:maximum-acceptance-lemma}, I first show that $\Gamma_G(b)$ in \eqref{eq:accept-wpone-constraint} is well-defined.
Let $\mathscr{G}(b)\coloneq\{a\ge \Xl \mid a<b \text{ and } |r([a,b]) - X| \le \pi^R, \forall X\in[a,b]\}$ be the set of left endpoints such that the manager's report under \hyperref[eq:reporting-problem]{$\mathrm{Report}(a,b)$} is always acceptable, conditional on the message $D=[a,b]$.
If $b \le -\pi^R$, then the manager's report is at most $r=0$, so there is no message that guarantees acceptance.
Thus I focus on the case $b > -\pi^R$.

\begin{lemma}\label{lem:reporting-acceptance}
    Suppose that $b > -\pi^R$. Then, for any $b>\Xl$, the set $\mathscr{G}(b)$ is nonempty, and $\min \mathscr{G}(b)$ exists.
\end{lemma}

\begin{proof}
    Consider a message $[a, b]$ for a fixed $b > -\pi^R$.
    From \eqref{eq:acceptance-probability} and \eqref{eq:acceptance-probability-large}, it is necessary that $b - a \le 2\pi^R$ to ensure that the induced report $r([a, b])$ is acceptable with probability one.
    Moreover, to ensure that the manager chooses the safe option (i.e., $r^* = a + \pi^R$), we require (see the proof of \zcref{lem:reporting-problem-solution})
    \begin{equation}
        \frac{\partial}{\partial r}\left[r(1-F_{[a,b]}(r-\pi^{R}))\right]_{r=a+\pi^{R}}\le 0,
    \end{equation}
    or equivalently,
    \begin{equation}\label{eq:no-gamble-condition}
        a+\pi^R \ge \frac{F(b)-F(a)}{f(a)}.
    \end{equation}
    I show that there exists a unique threshold $\Gamma_G(b)$ such that \eqref{eq:no-gamble-condition} holds if and only if $a\ge\Gamma_G(b)$.
    The right-hand side of \eqref{eq:no-gamble-condition} is strictly decreasing in $a$ by the log-concavity of $f$.
    In the limit $a\uparrow b$, the inequality \eqref{eq:no-gamble-condition} is satisfied due to $\pi^R>0$.
    In the limit $a \downarrow -\pi^R$, the inequality is violated due to the positivity of the right-hand side.
    Therefore, there is a unique $\Gamma_G(b)<b$ such that $\Gamma_G(b)+\pi^R=(F(b)-F(\Gamma_G(b)))/f(\Gamma_G(b))$, and we can take $\min \mathscr{G}(b)=\max\{\Gamma_G(b), b-2\pi^R\}<b$, proving the lemma.
\end{proof}

Define $\Gamma(b)\coloneq \min \mathscr{G}(b) = \max\{\Gamma_G(b), b-2\pi^R \}$.
\zcref{lem:maximum-acceptance-lemma} can now be restated as follows.

\renewcommand*{\thelemmap}{\ref*{lem:maximum-acceptance-lemma}$'$}   
\begin{lemmap}\label{lem:maximum-acceptance-prime}
    In an auditor-optimal equilibrium, any on-path message $[a,b]\subset [\Xl,\Xh]$ satisfies $a\ge \Gamma(b)$ for all $b>-\pi^R$.
    If $b\in(\Xl,-\pi^R]$, then it is without loss to set $a=\Xl$.
\end{lemmap}

\begin{proof}
    To prove the first part, take any $b>-\pi^R$. 
    Fix any auditor-optimal communication strategy $\sigma$.
    Toward a contradiction, suppose that $a < \Gamma(b)$.
    Assume without loss that $\Xl\le \min\{b-2\pi^R, \Gamma_G(b)\}$.
    \footnote{If $\Gamma(b)=\Xl$, then $a<\Gamma(b)$ would be impossible.}
    The hypothesis implies that either (i) $b-a\le 2\pi^R$ and $a<\Gamma_G(b)$ or (ii) $b-a>2\pi^R$.

    In case (i), $a<\Gamma_G(b)$ implies that $r^*\coloneq r([a,b])>a+\pi^R$.
    Construct an alternative auditor strategy $\sigma'$ that differs from $\sigma$ only on $[a, b]$ by splitting it into two intervals: $[a, r^*-\pi^R]$ and $[r^*-\pi^R, b]$.
    Under the new strategy $\sigma'$, the manager's report given the message $[r^*-\pi^R,b]$ remains at $r^*$, because
    the upper bound of the interval remains $b$ and $r^*$ is still feasible.
    \footnote{That is, $(r^*-\pi^R)+\pi^R\ge r^*$, so the risky option is in $[r^*-\pi^R,b]$.}
    Thus any $X\in[r^*-\pi^R,b]$ results in the same auditor payoff as under the original strategy.
    When $X\in[a, r^*-\pi^R]$, the manager's report changes from $r^*$ but is accepted with some positive probability. 
    Therefore, the auditor's expected payoff on $[a, b]$ is strictly higher under $\sigma'$, contradicting the optimality of the original strategy $\sigma$.
    In case (ii), the same splitting argument applies when $r^{*}>a+\pi^R$;  otherwise, split at $r^{*}+\pi^R<b$ instead.

    To prove the second part, suppose that $(\Xl,-\pi^R]$ is non-degenerate and fix any $b\in (\Xl,-\pi^R]$.
    For any $a<b$, the manager's report $r([a,b])$ is zero.
    Moreover, $|X-0|\ge \pi^R$ for all $X\le b$. 
    Therefore, any $a<b$ gives the same expected loss as $a=\Xl$.
\end{proof}

The above argument establishes a sufficient condition for the manager's report to be acceptable with probability one:

\begin{corollary}\label{cor:accept-wpone-cdn}
    Suppose that $\Xl \ge -\pi^R$.
    Then, in any auditor-optimal outcome, the manager's report is acceptable with probability one.
\end{corollary}

When this sufficient condition does not hold, then it is without loss to bundle all realizations $X< -\pi^R$ as a single message.

\begin{corollary}\label{cor:smallest-message}
    Suppose that $\Xl < -\pi^R$.
    Then, in any auditor-optimal equilibrium, it is without loss to let the smallest message be $[\Xl, -\pi^R]$, which induces an unacceptable report with probability one.
\end{corollary}

Consequently, when $\Xl \ge -\pi^R$, we can safely ignore the smallest message $[\Xl, -\pi^R]$ and focus instead on the partition that induces acceptable reports with probability one over the remaining support.

\subsection*{Proof of \zcref{lem:no-gambling-const}}
\begin{proof}
    Define $h(b)\coloneq \Gamma_G(b) - (b-2\pi^R)$. 
    I show that $\lim_{b\downarrow -\pi^R} h(b) > 0$ and that $h$ is monotonically decreasing in $b$.
    To establish the first claim, recall from \zcref{lem:reporting-acceptance} that $\Gamma_G(b)$ is the unique $a$ that solves $a+\pi^R = (F(b)-F(a))/f(a)$.
    Since $a\in(-\pi^R,b)$, the solution $\Gamma_G(b)$ tends to $-\pi^R$ as $b\downarrow -\pi^R$.
    Therefore $\lim_{b\downarrow -\pi^R} h(b) = 2\pi^R>0$.
    
    To prove that $h$ is decreasing, it suffices to show that $\partial \Gamma_G / \partial b < 1$. 
    Observe that
    \begin{equation}\label{eq:app_Gamma_G_derivative}
        \frac{\partial\Gamma_G(b)}{\partial b}=\frac{f(b)/f(a)}{1-\frac{\partial}{\partial a}\frac{F(b)-F(a)}{f(a)}}=\frac{f(b)}{2f(a)+(F(b)-F(a))\frac{f'(a)}{f(a)}}.
    \end{equation}
    Let $g(x)\coloneq \log f(x)$.
    Since $f$ is log-concave, $g'(x)\le g'(a)$ for all $x\in [a,b]$. 
    Therefore, 
    \begin{equation}
        \int_a^b g'(x)f(x)dx\le (F(b)-F(a))g'(a).
    \end{equation}
    Using $g'(x)=f'(x)/f(x)$, we can rewrite this as
    \begin{equation}
        f(b)\le f(a) + (F(b)-F(a)) \frac{f'(a)}{f(a)}.
    \end{equation}
    Adding $f(a)>0$ to the right-hand side of this inequality establishes that $\partial \Gamma_G / \partial b < 1$ by \eqref{eq:app_Gamma_G_derivative}.
\end{proof}

\subsection*{Proof of \zcref{thm:optimal-partition} and Off-Path Beliefs}
First I prove \zcref{thm:optimal-partition} by showing that the auditor does not benefit from deviating to any off-path message.
Then I discuss equilibrium refinements that justify such a belief.

Let $\mathfrak{M}(D)\coloneq\{X\mid \Mvag(X)\ni D\}$  denote the set of types who can also communicate $D$.
In the model, each message is a subset of the support of $X$.
As a useful consequence of this structure, we have $\mathfrak{M}(D) = D$.
\footnote{To see this, note that any type $X\in D$ is permitted to communicate this fact, so $D\subset \mathfrak{M}(D)$.
Conversely, if $X\notin D$, then truthful communication implies that type $X$ cannot use message $D$; thus $X\notin \mathfrak{M}(D)$.}

The manager's belief is \textit{wishful} if, for any off-path message $D$, the manager assigns probability one to $X = \sup\mathfrak{M}(D)=\sup{D}$.
That is, the manager interprets any off-path message in the way most favorable to himself.

\begin{proof}[Proof of \zcref{thm:optimal-partition}]
    Suppose that the manager's off-path belief is wishful.
    Let $[a,b]\in \mathcal{D^*}$ be a message in the optimal partition.
    Consider the auditor with realization $X\in [a,b]$.
    Take any off-path message $D \in \Mvag(X)$ such that $D\ne [a,b]$.
    On the equilibrium path, the difference between the manager's report and $X$ is $|r([a,b])-X|\le \pi^R$.
    If the auditor deviates to $D$, then the manager's wishful off-path belief $B(D)$ assigns probability one to $X=\sup D$.
    Therefore, the manager reports $r(D)=\sup D+\pi^R$.
    After the deviation, the difference between the report and $X$ is then
    \begin{align*}
        |r(D)-X| & =|\sup D+\pi^{R}-X|\\
                & =\sup D+\pi^{R}-X \ge\pi^{R},
    \end{align*}
    where the second line is from $D\ni X$.
    Therefore, for any message $D\in \mathcal{D^*}$ and $X\in D$, the auditor does not have a profitable deviation to any off-path message.
\end{proof}

\subsubsection*{Equilibrium Refinement}
A natural question is whether the assumption of the wishful off-path belief is ``reasonable.''
Here I show that such a belief survives the Grossman-Perry-Farrell criterion, a common refinement used in the truthful communication literature \citep{bertomeuVerifiableDisclosure2018,glodeVoluntaryDisclosureBilateral2018}.
\footnote{See \cite{aliDesignDisclosure2024} and \cite{titovaPersuasionVerifiableInformation2025} for a related discussion on when a sender-optimal outcome is achieved in a disclosure game.}
In particular, the technique developed in \cite{glodeVoluntaryDisclosureBilateral2018} applies to the current setting as well.

\begin{definition}
    A PBE $\langle \sigma, r, \alpha, B \rangle$ is a \emph{Grossman-Perry-Farrell equilibrium} (GPFE) if no non-empty interval $\self \subset [\Xl, \Xh]$ exists such that all types $X \in \self$ strictly benefit from deviating to the message $\self$, assuming the receiver's belief upon deviation, $B(\self)$, is the posterior distribution of $X$ conditional on $X \in \self$.
\end{definition}

Intuitively, the GPFE refinement rules out the following scenario.
Suppose type $X \in \self$ deviates to the message $\self$, and the manager believes that only types in $\self$ would do so. 
If the manager's belief is self-fulfilling---all types $\self$ indeed wish to deviate to the message $\self$---then we have found a pair of ``reasonable'' off-path deviation and belief: the auditor's message that she is in $\self$ is, in this sense, credible.

\begin{proposition}\label{prop:app_GPFE}
    An auditor-optimal equilibrium is a GPFE.
\end{proposition}

\begin{proof}
    Let $\sigma$ be an equilibrium communication strategy inducing the auditor-optimal partition.
    If there is a self-signaling set $\self$, an alternative communication strategy $\sigma'$ that differs from $\sigma$ only on $\self$ by assigning the message $\self$ to all $X\in \self$ strictly improves the auditor's expected payoff.
    This contradicts the optimality of $\sigma$, because $\sigma'$ can be supported by the same belief system $B$ in the original equilibrium.
\end{proof}

The result justifies the wishful off-path belief in the sense that an auditor-optimal equilibrium supported by such a belief survives the standard equilibrium refinement.
The GPFE criterion offers another justification for an auditor-optimal equilibrium by providing a partial converse of \zcref{prop:app_GPFE}.

\begin{proposition}\label{prop:app_GPFE_converse}
    Consider an equilibrium in which the acceptance probability is not maximized (and thus not auditor-optimal).
    Then, the equilibrium is not a GPFE.
\end{proposition}

\begin{proof}
    Let $\sigma$ be an equilibrium communication strategy under which the acceptance probability is not maximized.
    Then, there exists an on-path message $D=[a,b]$ with $b>-\pi^R$ such that the set $D^+ \cup D^-$ has positive measure, where
    \begin{equation}
        D^{+} \coloneq\{X\in D\mid X > r(D) + \pi^{R}\},\quad
        D^{-} \coloneq\{X\in D\mid X < r(D) - \pi^{R}\}.
    \end{equation}
    Suppose that $D^+$ is nondegenerate.
    Let $\self \coloneq (\inf{D^+}, b')$ for some $b'\in(\inf{D^+},b)$ such that $\inf D^+ \ge \Gamma(b')$.
    The manager's report in response to this message satisfies $|r(\self) - X|<\pi^R$ for all $X\in\self$.
    Since the auditor with realization $X\in D^+$ receives a constant payoff of $-(\pi^R)^2$ under $\sigma$, $\self$ is a self-signaling set.
    A similar argument applies if $D^-$ is nondegenerate.
\end{proof}

\subsection*{Proof of \zcref{prop:partition-DP}}
Assume without loss that $\Xl > -\pi^R$ (see \zcref{cor:accept-wpone-cdn, cor:smallest-message}).

\subsubsection*{The Equivalence of the Optimal Partition Problem and the Bellman Equation}
I first show that the problem \ref{eq:optimal-partition-problem} is equivalent to the Bellman equation \eqref{eq:DP-problem}.
To do so, it is convenient to index a partition from the right (``negative indexing''). 
Let $\mathscr{P}([\Xl,d_0])$ be the set of all partitions of $[\Xl,d_0]$.
A partition $\mathcal{D} \in \mathscr{P}([\Xl,d_0])$ with the negative indexing is written as
\begin{equation}
    \mathcal{D} = \bigcup_{i\ge 0} D_{-i}, \quad D_{-i}\coloneq [d_{-i-1},d_{-i}],
\end{equation}
where $i$ is a non-negative integer and $d_0\ge d_{-1}\ge d_{-2}\ge \cdots$.

For each message $D_{-i}=[d_{-i-1},d_{-i}]$, the constraint to ensure an acceptable report is
\begin{equation}\label{eq:acceptance-constraint}
    d_{-i-1}\in [\Gamma(d_{-i}),d_{-i}),
\end{equation}
Rewrite the problem \ref{eq:optimal-partition-problem} for the truncated support $[\Xl,d_0]$ as\footnote{When $\Xh<\infty$ we can set $d_0=\Xh$, and the conditioning becomes redundant.}
\begin{equation}\tag{$\mathrm{SP}$}\label{eq:SeqProb}
    \inf_{\mathcal{D} \in \mathscr{P}([\Xl,d_0])} \sum_{i\ge 0} \prob(X\in D_{-i}\mid X\le d_0) \ell(d_{-i-1},d_{-i}).
\end{equation}
Rewrite the Bellman equation \eqref{eq:DP-problem} as
\begin{equation}\tag{$\mathrm{BE}$}\label{eq:BellmanEq}
    L(d_0)=\inf_{d_{-1}\in[\Gamma(d_{0}),d_{0}]}\prob(X\ge d_{-1}\mid X\le d_{0})\ell(d_{-1},d_{0})+\prob(X\le d_{-1}\mid X\le d_{0})L(d_{-1}).
\end{equation}
I demonstrate that the objective function of \ref{eq:SeqProb} can be rewritten as the right-hand side of \ref{eq:BellmanEq}.
Let $\Lambda(\mathcal{D}) \coloneq \sum_{i\ge 0} \prob(X\in D_{-i}\mid X\le d_0) \ell(d_{-i-1},d_{-i})$ be the objective function.
Then,
\begin{align*}
    \Lambda(\mathcal{D})
    & =\prob(d_{-1}\le X\le d_{0}\mid X\le d_{0})\ell(d_{-1},d_{0})+\sum_{i\ge1}\prob(X\in D_{-i}\mid X\le d_{0})\ell(d_{-i-1},d_{-i})\\
    & =\prob(X\ge d_{-1}\mid X\le d_{0})\ell(d_{-1},d_{0})+\frac{\prob(X\le d_{-1})}{\prob(X\le d_{0})}\sum_{i\ge1}\prob(X\in D_{-i}\mid X\le d_{-1})\ell(d_{-i-1},d_{-i})\\
    & =\prob(X\ge d_{-1}\mid X\le d_{0})\ell(d_{-1},d_{0})+\prob(X\le d_{-1}\mid X\le d_{0})\Lambda(d_{-1}),
\end{align*}
where the first equality follows from isolating the initial term, the second from the monotonicity of the sequence \{$d_i$\} and Bayes' rule, and the third from the definition of $\Lambda$.
Therefore, the solution sets to \ref{eq:BellmanEq} and \ref{eq:SeqProb} coincide, following the arguments in \citet[Chap.~4]{stokeyRecursiveMethodsEconomic1989}.

\subsubsection*{Contraction Property}
Now I show that there is a unique fixed point of the Bellman operator associated with \ref{eq:BellmanEq}.
Let $\mathcal{C}_B([\Xl,\Xh])$ denote the space of bounded continuous functions on $[\Xl,\Xh]$, endowed with the supremum norm $\|\cdot\|_{\infty}$.
Define the Bellman operator $T:\mathcal{C}_B([\Xl,\Xh]) \to \mathcal{C}_B([\Xl,\Xh])$ by 
\begin{equation}
    TL(d_0) \coloneq \min_{d_{-1}\in[\Gamma(d_{0}),d_{0}]}\prob(X\ge d_{-1}\mid X\le d_{0})\ell(d_{-1},d_{0})+\prob(X\le d_{-1}\mid X\le d_{0})L(d_{-1}),
\end{equation}
and let $\lambda(d_0)$ denote the minimizer in the above expression.
It suffices to show that $T$ is a contraction mapping \citep[cf.][]{stokeyRecursiveMethodsEconomic1989}.
To do so, I first establish a uniform upper bound on the probability term.

\begin{lemma}\label{lem:contraction-bound}
    There exists $\beta\in (0,1)$ such that $\prob(X\le \lambda(d_0)\mid X\le d_0)\le \beta$ for all $d_0\in(\Xl,\Xh]$.
\end{lemma}
\begin{proof}
    The result holds trivially if $\Xh < \infty$, so assume $\Xh=\infty$.
    Toward a contradiction, suppose that 
    \begin{equation}
        \lim_{d_0\to \infty}\prob(X\le \lambda(d_0)\mid X\le d_0)=1.
    \end{equation}
    Take any $\varepsilon>0$. 
    There exists $\overline{d}$ such that $\prob(X\le \lambda(b)\mid X\le b) \ge 1-\varepsilon$ for all $b\ge \overline{d}$.
    Since $\varepsilon>0$ is arbitrary and the density vanishes at infinity, we may take $\overline{d}$ such that $b-\lambda(b)<\pi^R$ and $\lambda(b)>\overline{d}$ for all $b>\overline{d}$.
    In this case, the manager's report $r([\lambda(b),b]) = \lambda(b)+\pi^R$ strictly exceeds $[\lambda(b),b]$ for any $b>\overline{d}$.
    Then the auditor's conditional expected loss, $\ell(\lambda(b), b) = \E[(\lambda(b)+\pi^R-X)^2\mid X\in [\lambda(b), b]]$, becomes smaller by merging the message $[\lambda(b), b]$ with the next adjacent interval to the left.
    This contradicts the optimality of $\lambda(b)$.
\end{proof}

Now I prove that $T$ is a contraction mapping, which effectively proves \zcref{prop:partition-DP}.

\begin{proof}[Proof of \zcref{prop:partition-DP}]
    Let $L_1, L_2 \in \mathcal{C}_B([\Xl, \Xh])$ be arbitrary, and let $\lambda_1(d_0)$ and $\lambda_2(d_0)$ denote their corresponding minimizers.
    Then for $k \in {1, 2}$ and $k' \ne k$,
    \begin{equation}
        (TL_k)(d_0) \le \prob(X\ge \lambda_{k'}(d_0)\mid X\le d_0)\ell(\lambda_{k'}(d_0), d_0) + \prob(X\le \lambda_{k'}(d_0)\mid X\le d_0)L_{k}(\lambda_{k'}(d_0)).
    \end{equation}
    Therefore,
    \begin{align}
        |TL_{1}-TL_{2}|(d_{0}) & \le\max_{k\in\{1,2\}}\prob(X\le\lambda_{k}(d_{0})\mid X\le d_{0})\|L_{1}-L_{2}\|_{\infty}\\
        & \le\beta\|L_{1}-L_{2}\|_{\infty},
    \end{align}
    where the second line is from \zcref{lem:contraction-bound}.
    Taking the supremum over $d_0$ gives $\|TL_{1}-TL_{2}\|_{\infty} \le \beta\|L_{1}-L_{2}\|_{\infty}$, so $T$ is a contraction. 
\end{proof}

\subsection*{Proof of \zcref{prop:compara-pi-uniform}}
\begin{proof}
    When $\pi^R < \Xl$, a uniform partition is optimal. 
    Let $\Xl>0$ and define 
    \begin{equation}
        \ell:\R_{++} \times (0,\Xl) \ni (\Delta,\pi^R)\mapsto (\pi^R)^2 - \pi^R \Delta + (\Delta)^2/3\in \R,
    \end{equation}
    which gives the auditor's expected loss when each interval has length $\Delta$.
    The feasible set of interval lengths is
    \begin{equation}
        S(\pi^R)\coloneq \{\Delta>0\mid\Delta\le2\pi^R,\,\Delta=(\Xh-\Xl)/N\text{ for some }N\in\mathbb{N}\}.
    \end{equation}
    With this notation, the optimal partition problem \eqref{eq:uniform-uniform-partition-problem} can be rewritten as 
    \begin{equation}\label{eq:app_uniform-uniform-partition-problem}
        \min_{\Delta\in S(\pi^R)} \ell(\Delta\, ; \pi^R).
    \end{equation}

    It suffices to show that the solution to \eqref{eq:app_uniform-uniform-partition-problem} is weakly increasing in $\pi^R$.
    Results from Monotone Comparative Statics offer a simple way to verify this.
    Note $\ell$ has strictly decreasing differences in $(\Delta, \pi^R)$, since $\frac{\partial^2\ell}{\partial\pi^{R}\partial\Delta}=-1$.
    The feasible set $S(\pi^R)$ is increasing in $\pi^R$ in the strong set order.
    Therefore, every selection from the solution set of \eqref{eq:app_uniform-uniform-partition-problem} is weakly increasing in $\pi^R$ \citep{topkisMinimizingSubmodularFunction1978,milgromMonotoneComparativeStatics1994}.
    \footnote{Since $\ell(\cdot\,;\pi^R)$ is defined on a chain, it is trivially quasi-supermodular.}
\end{proof}

\begin{remark}
    In the proof, the monotonicity is established for an arbitrary selection of the solution correspondence, because the problem \eqref{eq:app_uniform-uniform-partition-problem} may have multiple solutions.
    In particular, when $\Delta^\mathrm{ideal}=1.5\pi^R$ is in the middle of some interval $[(\Xh-\Xl)/N, (\Xh-\Xl)/N']$, then the auditor is indifferent between $N$ and $N'$ partitions as long as they are both feasible.
    \qedhere
\end{remark}

\section{Other Applications}\label{appsec:other-applications}
I discuss several additional applications of my model.
For each setting, I sketch how the model applies and outline the potential insights it offers.
The settings I consider are bank stress testing, capital budgeting, environmental regulation, and patent examination.

\paragraph{Bank Stress Testing}
Banks in the U.S. are subject to periodic stress tests by the Federal Reserve.
The tests evaluate whether banks can withstand adverse economic scenarios.
An important policy question is how transparent the regulator should be about the models it uses in these tests \citep{leitnerModelSecrecyStress2023}.
\footnote{Another issue of transparency is whether the regulator should disclose the test's results. The key tradeoff in that context is distinct from the transparency vis-à-vis banks I discuss here. See \cite{goldsteinShouldBanksStress2014} and \cite{goldsteinStressTestDisclosure2022} for reviews on stress tests disclosure.}
The regulator faces the gatekeeping expert's dilemma: being too transparent may lead banks to game the system, while being too opaque may leave them unprepared for adverse scenarios as they do not understand the regulator's concerns.
Indeed, the regulator has repeatedly raised this concern \citep{tarulloSpeechGovernorTarullo2016,barrSpeechGovernorBarr2025}.

To apply my framework to this setting, consider a bank preparing for an upcoming stress test.
The bank must choose its level of risk-taking, denoted by $r$.
A higher $r$ may reflect actions such as holding less capital, lending to riskier borrowers, or extending the maturity of long-duration bond holdings.
The regulator, drawing on internal stress test models and macroeconomic forecasts, knows the bank's optimal risk level, denoted by $X$.
The bank, by contrast, has an incentive to take on more risk (i.e., a higher $r$) in order to boost returns, possibly due to the ``too-big-to-fail'' problem \citep{strahanTooBigFail2013}.

Before the stress test, the regulator may choose to communicate aspects of its stress tests---such as the test scenarios, model parameters, or the timing of the test.
By doing so, the regulator would like to discipline the bank's choice of $r$.
Ultimately, the regulator either passes or fails the bank based on its risk profile $r$ and the regulator's private information $X$.

\cite{leitnerModelSecrecyStress2023} analyze the tradeoff between transparency and secrecy in stress testing.
They show that regulators benefit from disclosing some information when they can optimally set the test threshold.
My framework complements their findings by examining a setting in which the regulator cannot commit to a pass/fail threshold in advance.
In this environment, vague communication can serve as a disciplinary tool.
The comparative statics in my model yield new insights into how the regulator's independence and expertise shape its degree of secrecy.
The model also offers a potential explanation for why most banks pass the stress tests: 
\footnote{For example, in the 2025 tests, all 22 participating banks passed.
See \url{https://www.federalreserve.gov/supervisionreg/stress-tests-capital-planning.htm} for the past test results.}
the regulator chooses an appropriate level of ``model secrecy'' to ensure that banks select appropriate risk levels without resorting to costly failures of the tests.

\paragraph{Bottom-Up Capital Budgeting}
Consider a capital budgeting process within a firm.
While budgetary practices vary across organizations, one common approach is bottom-up budgeting, in which operating units prepare capital investment proposals for approval by a central capital budgeting committee \citep{shimBudgetingBasics2005}.
The budgeting committee reviews these proposals and decides whether to approve or reject them.
It also sets key approval parameters, such as hurdle rates \citep{rossCapitalBudgetingPractices1986}.

To model this situation, consider a large manufacturing company with a bottom-up budgeting approach, where division managers routinely propose capital investment projects.
A manager is offered a potential technology investment by a vendor.
The new technology would replace aging equipment and improve efficiency.
The manager requests resources $r$ from the company's headquarters to fund the project. 
Because the project requires a large capital outlay, it must be approved by the capital budgeting committee \citep{rossCapitalBudgetingPractices1986}.

The committee comprises experts, including current facility managers with domain expertise in the relevant technology and operations.
When the manager submits a proposal, the committee can advise on what it deems to be the optimal investment level, denoted by $X$. 
\cite{shimBudgetingBasics2005} note that one function of a budgeting committee is to provide advice.
Suppose that if the committee rejects the initial proposal, submitting a revised one is too costly.
For example, the company has limited capital to allocate; if the project is rejected, the committee may instead fund competing proposals. 
Alternatively, the vendor may walk away and go to competitors.

The key premise is the agency problem: the headquarters cannot execute the project itself, and the manager derives private benefits from the allocated resources \citep{steinInternalCapitalMarkets1997}.
The literature has extensively examined how information asymmetry and agency problems influence capital budgeting \citep{baimanValuePrivatePreDecision1991,gervaisOverconfidenceCompensationContracts2011,almazanFirmInvestmentStakeholder2017}.
For example, in a principal-agent framework, \cite{baimanValuePrivatePreDecision1991} show that giving the manager more pre-decision information may exacerbate the agency problem, making the principal (e.g., headquarters) worse off.
My model complements this insight by showing that the capital budgeting committee can optimally utilize vague communication to mitigate the agency problem in a setting where the principal has only veto power.

\paragraph{Environmental Regulation}
Consider a firm that wants to open a production facility.
The firm must file a permit application that specifies the details of the production plan, including an allowable emissions cap $r$.
The environmental regulator (e.g., the EPA) reviews the application and decides to approve or reject it.
\footnote{See \cite{laffontPollutionPermitsCompliance1996} for a model of pollution permit allocation.}
If the regulator rejects it, the facility cannot be opened and the opportunity is lost.

The regulator is a gatekeeping expert, who can determine the optimal level of emissions cap $X$.
The firm benefits from looser limits (higher $r$), because it can produce more products without costly abatement.
The regulator may communicate guidance when the firm submits its application. 
This guidance can be precise or intentionally vague (e.g., ``our modeling suggests your cap should be around 1-1.5\unit{\tonne} per day'').
After seeing the plant's proposal, the regulator either approves the application as filed (allocating the requested cap $r$) or rejects it. 
The regulatory review is a time-consuming process, and there is effectively only one chance to get approval.
Going through the process again would be too costly for the firm. 

In practice, environmental regulators often provide vague guidance.
For example, take the enforcement of the Clean Water Act of 1972.
A persistent challenge has been determining whether a given property falls under the Act's jurisdiction.
In the recent Supreme Court case \textit{Sackett v. EPA}, the vagueness of the EPA's jurisdictional test was challenged.
During oral arguments, Justice Sotomayor captured the crux of the problem, asking, ``But is there another test that could be more precise and less open-ended than the [...] test that you [the EPA] use?''
The EPA's vague communication may be a strategic choice: by avoiding a precise test, the agency preserves its gatekeeping discretion while still influencing property owners' behavior.

More broadly, \cite{blanchardPortfolioEconomicPolicies2023} note that opaque environmental policies are frequently favored over clear, price-based regulations such as a carbon tax.
My model provides a framework for understanding the strategic use of vague guidance in environmental regulation.

\paragraph{Patent Examination}
Consider a company that has developed a new technology.
The company files an application, which is characterized by the breadth claimed, denoted by $r$. 
A higher $r$ corresponds to an application that claims a broader scope and novelty of protection for the technology.
A patent examiner, who is an expert in the relevant technology and patent law, reviews the application and decides whether to approve it.

After searching prior art and applying relevant standards, the examiner can identify the legally allowable scope $X$ of the claimed invention.
The examiner would like the claim $r$ to be as close as possible to $X$.
The company prefers broader protection (a higher $r$), as broader claims raise expected licensing revenue and deter entry. 
When the firm submits its initial claims, the examiner may respond with guidance.
In practice, patent examiners usually send back the application, laying out issues in the application and suggesting amendments \citep{lemleyPatentOfficeRubber2008}.

In this context, my model highlights the importance of patent examiners' expertise.
\cite{lemleyExaminerCharacteristicsPatent2012} document significant variation in examiners' expertise at the U.S. Patent and Trademark Office (USPTO).
They find that experienced examiners are more likely to grant patents.
My model suggests that this may be partly because an experienced examiner can effectively communicate with the applicant in the initial rejection, thereby ensuring an acceptable claim scope in the next round.

\end{document}